%% file: arxiv_version.tex
\PassOptionsToPackage{unicode}{hyperref}
\PassOptionsToPackage{hyphens}{url}
\PassOptionsToPackage{dvipsnames,svgnames,x11names}{xcolor}
\documentclass[
  12pt]{article}

\usepackage{amsmath,amssymb}
\usepackage{iftex}
\ifPDFTeX
  \usepackage[T1]{fontenc}
  \usepackage[utf8]{inputenc}
  \usepackage{textcomp} 
\else 
  \usepackage{unicode-math}
  \defaultfontfeatures{Scale=MatchLowercase}
  \defaultfontfeatures[\rmfamily]{Ligatures=TeX,Scale=1}
\fi
\usepackage{lmodern}
\ifPDFTeX\else  
\fi
\IfFileExists{upquote.sty}{\usepackage{upquote}}{}
\IfFileExists{microtype.sty}{
  \usepackage[]{microtype}
  \UseMicrotypeSet[protrusion]{basicmath} 
}{}
\makeatletter
\@ifundefined{KOMAClassName}{
  \IfFileExists{parskip.sty}{%
    \usepackage{parskip}
  }{
    \setlength{\parindent}{0pt}
    \setlength{\parskip}{6pt plus 2pt minus 1pt}}
}{
  \KOMAoptions{parskip=half}}
\makeatother
\usepackage{xcolor}
\makeatletter
\ifx\paragraph\undefined\else
  \let\oldparagraph\paragraph
  \renewcommand{\paragraph}{
    \@ifstar
      \xxxParagraphStar
      \xxxParagraphNoStar
  }
  \newcommand{\xxxParagraphStar}[1]{\oldparagraph*{#1}\mbox{}}
  \newcommand{\xxxParagraphNoStar}[1]{\oldparagraph{#1}\mbox{}}
\fi
\ifx\subparagraph\undefined\else
  \let\oldsubparagraph\subparagraph
  \renewcommand{\subparagraph}{
    \@ifstar
      \xxxSubParagraphStar
      \xxxSubParagraphNoStar
  }
  \newcommand{\xxxSubParagraphStar}[1]{\oldsubparagraph*{#1}\mbox{}}
  \newcommand{\xxxSubParagraphNoStar}[1]{\oldsubparagraph{#1}\mbox{}}
\fi
\makeatother

\usepackage{longtable,booktabs,array}
\usepackage{calc} 
\usepackage{etoolbox}
\makeatletter
\patchcmd\longtable{\par}{\if@noskipsec\mbox{}\fi\par}{}{}
\makeatother
\IfFileExists{footnotehyper.sty}{\usepackage{footnotehyper}}{\usepackage{footnote}}
\makesavenoteenv{longtable}
\usepackage{graphicx}
\makeatletter
\def\maxwidth{\ifdim\Gin@nat@width>\linewidth\linewidth\else\Gin@nat@width\fi}
\def\maxheight{\ifdim\Gin@nat@height>\textheight\textheight\else\Gin@nat@height\fi}
\makeatother
\setkeys{Gin}{width=\maxwidth,height=\maxheight,keepaspectratio}
\makeatletter
\def\fps@figure{htbp}
\makeatother

\makeatletter
\@ifpackageloaded{caption}{}{\usepackage{caption}}
\AtBeginDocument{%
\ifdefined\contentsname
  \renewcommand*\contentsname{Table of contents}
\else
  \newcommand\contentsname{Table of contents}
\fi
\ifdefined\listfigurename
  \renewcommand*\listfigurename{List of Figures}
\else
  \newcommand\listfigurename{List of Figures}
\fi
\ifdefined\listtablename
  \renewcommand*\listtablename{List of Tables}
\else
  \newcommand\listtablename{List of Tables}
\fi
\ifdefined\figurename
  \renewcommand*\figurename{Figure}
\else
  \newcommand\figurename{Figure}
\fi
\ifdefined\tablename
  \renewcommand*\tablename{Table}
\else
  \newcommand\tablename{Table}
\fi
}
\@ifpackageloaded{float}{}{\usepackage{float}}
\floatstyle{ruled}
\@ifundefined{c@chapter}{\newfloat{codelisting}{h}{lop}}{\newfloat{codelisting}{h}{lop}[chapter]}
\floatname{codelisting}{Listing}

\makeatother
\makeatletter
\@ifpackageloaded{caption}{}{\usepackage{caption}}
\@ifpackageloaded{subcaption}{}{\usepackage{subcaption}}
\makeatother

\ifLuaTeX
  \usepackage{selnolig}  
\fi
\usepackage[]{natbib}
\usepackage{bookmark}

\IfFileExists{xurl.sty}{\usepackage{xurl}}{} 
\hypersetup{
  pdftitle={Title},
  pdfauthor={Author 1; Author 2},
  pdfkeywords={3 to 6 keywords, that do not appear in the title},
  colorlinks=true,
  linkcolor={blue},
  filecolor={Maroon},
  citecolor={Blue},
  urlcolor={Blue},
  pdfcreator={LaTeX via pandoc}}

\newcommand{\anon}{1}

\usepackage{mathtools}
\mathtoolsset{showonlyrefs} 
\usepackage{booktabs}
\usepackage{multirow}
\usepackage{algorithm}
\usepackage{algpseudocode}
\usepackage{tabularx}
\usepackage{longtable}
\usepackage{amsthm}
\newtheorem{remark}{Remark}
\newtheorem{proposition}{Proposition}
\usepackage{xspace}

\usepackage{listings}
\usepackage{xcolor}

\definecolor{codegreen}{rgb}{0,0.6,0}
\definecolor{codegray}{rgb}{0.5,0.5,0.5}
\definecolor{codepurple}{rgb}{0.58,0,0.82}
\definecolor{backcolour}{rgb}{0.95,0.95,0.92}

\lstdefinestyle{mystyle}{
    backgroundcolor=\color{backcolour},   
    commentstyle=\color{codegreen},
    keywordstyle=\color{magenta},
    numberstyle=\tiny\color{codegray},
    stringstyle=\color{codepurple},
    basicstyle=\ttfamily\footnotesize,
    breakatwhitespace=false,         
    breaklines=true,                 
    captionpos=b,                    
    keepspaces=true,                 
    numbers=left,                    
    numbersep=5pt,                  
    showspaces=false,                
    showstringspaces=false,
    showtabs=false,                  
    tabsize=2,
    frame=single
}
\newcommand{\fone}{\widehat{f}_1}
\newcommand{\ftwo}{\widehat{f}_2}

\newcommand{\boneopt}{h^*_1}
\newcommand{\btwoopt}{h^*_2}

\newcommand{\ahat}{\widehat{a}}
\newcommand{\bhat}{\widehat{b}}

\newcommand{\bw}{h}
\newcommand{\xbar}{\overline{x}}
\newcommand{\vbar}{\overline{v}}

\newcommand{\betareference}{h_{\mathrm{ref}}}
\newcommand{\hheur}{h_{\mathrm{heur}}}

\newcommand{\rot}{Beta Reference Rule\xspace}
\newcommand{\blscv}{Beta LSCV estimator\xspace}
\newcommand{\bise}{Beta ISE-optimal estimator\xspace}
\newcommand{\oracle}{Beta Oracle estimator\xspace}
\newcommand{\lsilv}{Logit-Silverman estimator\xspace}
\newcommand{\llscv}{Logit LSCV estimator\xspace}
\newcommand{\lise}{Logit ISE-optimal estimator\xspace}
\newcommand{\rsilv}{Reflection-Silverman estimator\xspace}
\newcommand{\rlscv}{Reflection LSCV estimator\xspace}
\newcommand{\rise}{Reflection ISE-optimal estimator\xspace}

\newcommand{\rott}{Beta (Ref)\xspace}
\newcommand{\blscvt}{Beta (LSCV)\xspace}
\newcommand{\biset}{Beta (ISE)\xspace}
\newcommand{\oraclet}{Beta (Oracle)\xspace}
\newcommand{\lsilvt}{Logit (Silverman)\xspace}
\newcommand{\llscvt}{Logit (LSCV)\xspace}
\newcommand{\liset}{Logit (ISE)\xspace}
\newcommand{\rsilvt}{Reflect (Silverman)\xspace}
\newcommand{\rlscvt}{Reflect (LSCV)\xspace}
\newcommand{\riset}{Reflect (ISE)\xspace}

\newcommand{\NT}{\mathcal{NT}}

\begin{document}

\def\spacingset#1{\renewcommand{\baselinestretch}%
{#1}\small\normalsize} \spacingset{1}


\if1\anon
{
  \title{\bf A Fast, Closed-Form Bandwidth Selector for the beta kernel Density Estimator}
  \author{Johan Hallberg Szabadv\'ary
    \hspace{.2cm}\\
    Department of Mathematics, Stockholm University \\
    and\\
    Department of Computing, Jönköping School of Engineering\\
    }
  \maketitle
} \fi

\if0\anon
{
  \bigskip
  \bigskip
  \bigskip
  \begin{center}
    {\LARGE\bf A Fast, Closed-Form Bandwidth Selector for the beta kernel Density Estimator}
\end{center}
  \medskip
} \fi

\bigskip
\begin{abstract}
     The beta kernel estimator offers a theoretically superior alternative to the Gaussian kernel for unit interval data, eliminating boundary bias without requiring reflection or transformation. However, its adoption remains limited by the lack of a reliable bandwidth selector, and practitioners currently rely on computationally expensive iterative optimization methods that are prone to instability. We derive the ``\rot'', a fast, closed-form bandwidth selector, based on the unweighted asymptotic mean integrated squared error (AMISE) of a beta reference distribution. To address boundary integrability issues, we introduce a principled heuristic for U-shaped and J-shaped distributions. By employing a method-of-moments approximation, we reduce the bandwidth selection complexity from iterative optimization to $\mathcal{O}(1)$. Extensive Monte Carlo simulations demonstrate that our rule matches the accuracy of numerical optimization while delivering a speedup of over 35,000 times. Real-world validation on socioeconomic data shows that it avoids the ``vanishing boundary'' and ``shoulder'' artifacts common to Gaussian-based methods. We provide a comprehensive, open-source Python package to facilitate the immediate adoption of the beta kernel as a drop-in replacement for standard density estimation tools.
\end{abstract}

\noindent%
\vfill

\noindent
{\it To appear in Journal of Computational and Graphical Statistics}
\vfill
\newpage

\section{Introduction}
    Kernel density estimation \citep{rosenblatt1965remarks,parzen1962estimation} is a nonparametric statistical method used to estimate the probability density function of a random variable based on a finite data sample.
    For a univariate random sample $X_1, \dots, X_n$ drawn from an unknown density $f$ with support in the unit interval $[0,1]$, the standard Gaussian kernel density estimator is theoretically ill-suited for the following reasons. Because the Gaussian kernel assumes an unbounded support $(-\infty, \infty)$, it suffers from a severe ``boundary bias'' near endpoints. In these regions, the probability mass ``leaks'' outside the valid domain, and the bias of the estimator degrades to order $\mathcal{O}(h)$, rather than the standard $\mathcal{O}(h^2)$ convergence rate achieved in the interior \citep{wand1994kernel}.

    Practitioners often attempt to mitigate this bias using ad hoc corrections; however, these introduce significant theoretical artifacts. The reflection method (see, e.g., \citet{karunamuni2005generalized}), which mirrors data across the boundaries to correct the probability mass, enforces an artificial symmetry constraint. As noted by \citet{schuster1985incorporating} and \citet{cowling1996pseudodata}, this technique forces the derivative of the estimated density to vanish at the boundaries ($\hat{f}'(0)=\hat{f}'(1)=0$). Consequently, for distributions with nonzero boundary slopes, such as exponential or power-law distributions, reflection introduces a systematic ``shoulder'' artifact that misrepresents the true shape of the data.

    Other boundary correction techniques, such as the linear boundary kernel proposed by \citet{jones1993simple}, successfully reduce bias but often produce density estimates that are negative near the boundaries, thus violating the fundamental properties of the probability density function.

    Alternatively, transformation methods (for example, logit or probit) map the unit interval to the real line, apply a standard Gaussian KDE, and map the result back. Although this ensures correct support, it has two critical flaws. First, the transformation is undefined for data points exactly at the boundaries ($0$ or $1$), necessitating arbitrary data adjustment. Second, the interplay between the transformation Jacobian and light tails of the Gaussian kernel typically forces the estimated density to vanish at the boundaries ($\hat{f}(0)=0$) \citep{geenens2014probit}. This makes transformation methods particularly ill-suited for estimating distributions that are nonzero or unbounded at the endpoints, such as uniform or U-shaped beta distributions.

    A more theoretically sound approach is to use a kernel function whose support naturally matches that of the data. \citet{chen1999beta} proposed the beta kernel estimator, which replaces the Gaussian kernel functions with Beta densities. Unlike reflection, the beta kernel does not impose an artificial derivative constraint. Unlike transformations, it operates directly in the native data space. It is strictly non-negative, free from boundary bias (achieving the optimal $\mathcal{O}(h^2)$ bias everywhere), and possesses natural adaptivity; the variance of the kernel decreases as the estimation point moves toward the boundaries, automatically reducing smoothing, where the data are naturally denser. Theoretically, the beta kernel is the superior estimator for unit-interval data. This approach inspired a broader class of asymmetric kernel estimators, including gamma kernels for semi-infinite support \citep{chen2000probability} and inverse Gaussian kernels \citep{scaillet2004density}, all of which share the property of matching the kernel support to the data domain.

    However, despite its competitive performance and attractive properties, the beta kernel has not gained traction among practitioners. The primary obstacle is the lack of a simple, closed-form bandwidth-selection rule. The popularity of the Gaussian kernel is due, in no small part, to the availability of reliable plug-in bandwidth selectors, such as Silverman’s rule \citep{silverman2018density} or the Sheather–Jones solve-the-equation method \citep{sheather1991reliable}, which provides an immediate, data-driven bandwidth. In contrast, users of the beta kernel are currently forced to rely on numerical optimization methods, such as least squares cross-validation (LSCV) \citep{rudemo1982empirical}. LSCV is not only computationally expensive and scales poorly with the sample size but is also notoriously unstable, often producing highly variable bandwidths that result in undersmoothed estimates \citep{hall1987kullback}. \citet{geenens2014probit} identified this lack of a simple bandwidth selector as a critical gap that effectively disqualifies the beta kernel from routine use.

    Although \citet{hirukawa2010nonparametric} derived analytical bandwidths for beta-based estimators, their approach focuses on multiplicative bias correction and relies on minimizing a weighted MISE (AWMISE) to ensure convergence. This results in computationally intensive expressions involving polygamma functions that do not yield rapid and transparent rules of thumb.

    In addition to standard global selectors, nonparametric statistics offer a rich array of advanced bandwidth selection strategies. For instance, spatially adaptive methods—such as Lepski’s method \citep{lepski1997optimal} and the data-driven variable bandwidth approach of \citet{fan1995data-driven}, allow the bandwidth to vary locally to balance bias and variance across different density regions. Furthermore, recent bootstrap-based procedures \citep{Liu02082024} provide highly effective, resampling-driven AMISE approximations. However, applying these advanced classes of selectors to the beta kernel presents significant conceptual and technical obstacles. First, the beta kernel is inherently spatially adaptive; its shape and effective variance naturally adjust as the evaluation point approaches the boundaries. Imposing a locally varying bandwidth on an already locally varying asymmetric kernel introduces severe analytical complexity. Second, the primary focus of this study is computational efficiency. Spatially adaptive and bootstrap-based selectors rely on exhaustive local grid searches or intensive resampling, making them computationally far more expensive than even global LSCV. Therefore, while these advanced methods offer high theoretical precision, there remains a critical need for a fast, closed-form, global $\mathcal{O}(1)$ rule of thumb for bounded, asymmetric kernels.

    In this study, we address this computational bottleneck by deriving a fast closed-form rule of thumb for the beta kernel bandwidth. Analogous to Silverman's rule (also known as the Gaussian reference rule), we derive the optimal bandwidth by minimizing the Asymptotic Mean Integrated Squared Error (AMISE) of a beta reference distribution. Our derivation yields a simple analytical formula based on the method of moments estimates of the data parameters. Furthermore, we identify the domain of applicability for this approximation and propose a principled heuristic fallback for ``hard'' (U-shaped or J-shaped) distributions, where the asymptotic approximation breaks down.

    Our contribution allows the beta kernel to be used with the same computational ease as the Gaussian kernel, reducing the cost from iterative optimization to $O(1)$ while retaining its superior boundary properties. We provide a fully documented, open-source Python package that implements the estimator and bandwidth selection rules, thereby making the beta kernel a drop-in replacement for the Gaussian KDE in modern data science workflows.

\section{The beta kernel}\label{sec:betaKernel}
    \citet{chen1999beta} proposed an interesting kernel that uses beta densities. It is free from boundary bias and achieves an optimal rate of convergence for the mean integrated squared error. An important feature is that the support of the kernel functions matches the data support. It is also interesting that different amounts of smoothing are allocated by naturally varying the kernel shape without explicitly changing the value of the smoothing bandwidth. For a sample $x_1,\dots,x_n$ from an unknown distribution $f$ in the unit interval, \citet{chen1999beta} proposed two  estimators,
    \begin{equation}
        \fone(x) = \frac{1}{n}\sum_{i=1}^nK_{x/\bw+1, (1-x)/\bw+1}(x_i),
    \end{equation}
    where $K_{a,b}$ is the density function of a $\text{beta}(a,b)$ distribution, and
    \begin{equation}
        \ftwo(x) = \frac{1}{n}\sum_{i=1}^nK^*_{x, \bw}(x_i),
    \end{equation}
    where $K^*_{x,\bw}$ are boundary beta kernels defined as
    \begin{equation}
        K^*_{x, \bw}(t) = 
        \begin{cases}
            K_{x/\bw, (1-x)/\bw}(t) & \text{if $x\in[2\bw, 1-2\bw],$}\\
            K_{\rho(x,\bw), (1-x)/\bw}(t) & \text{if $ x\in[0,2\bw)$}\\
            K_{x/\bw, \rho(1-x,\bw)}(t) & \text{if $ x\in(1-2\bw,1]$}.
        \end{cases}
    \end{equation}
    Here, $\rho(x, \bw) = 2\bw^2 + 2.5 - \sqrt{4\bw^4 + 6\bw^2 + 2.25 - x^2 - x/\bw}$.

    \citet{chen1999beta} derived the optimal bandwidths, that minimise the \emph{mean integrated squared error} (MISE), for $\fone$ and $\ftwo$. For an unknown density $f$, with support $[0,1]$, they are
    \begin{equation}
        \boneopt = \frac{\bigg(\frac{1}{2\sqrt{\pi}}\int_0^1\frac{f(x)}{\sqrt{x(1-x)}}dx\bigg)^{2/5}}{4^{2/5}\bigg(\int_0^1((1-2x)f'(x) + \frac{1}{2}x(1-x)f''(x))^2dx\bigg)^{2/5}}n^{-2/5}
    \end{equation}
    and
    \begin{equation}\label{eq:MISEbandwidth}
        \btwoopt = \frac{\bigg(\frac{1}{2\sqrt{\pi}}\int_0^1\frac{f(x)}{\sqrt{x(1-x)}}dx\bigg)^{2/5}}{\bigg(\int_0^1(x(1-x)f''(x))^2dx\bigg)^{2/5}}n^{-2/5}
    \end{equation}
    respectively. \citeauthor{chen1999beta} concluded that $\ftwo$ achieves a lower optimal MISE, and a smaller optimal bandwidth, and is therefore recommended over $\fone$. From this point onward, all references to the beta kernel refer to estimator $\ftwo$.

    \begin{remark}[Mass Conservation]
        It is a known property of asymmetric kernel estimators, including the beta kernel $\ftwo$, that they do not strictly preserve unit probability mass in finite samples \citep{jones1990variable,chen1999beta}. Although the estimator represents a valid probability density with respect to the data point $t$, it is not necessarily normalized with respect to the evaluation point $x$.  In Appendix \ref{app:AME}, we derive the exact asymptotic deviation, which shows that the total probability mass converges to unity at a rate of $O(h)$.
        Although post-hoc renormalization is possible \citep {jones1996aSimple}, we evaluate the estimator in its canonical, unnormalized form. This choice allows us to isolate the performance gains strictly attributable to the boundary-adaptive shape of the beta kernel. Furthermore, this provides a conservative assessment of performance, as any deviation from the unit mass strictly penalizes the estimator under the MISE metric used in the experiments.
    \end{remark}

    An interesting feature of the beta kernel is that its shape varies naturally (because $x$ determines the parameters of each individual beta density), which implies that the amount of smoothing varies according to the position where the density is estimated without explicitly changing the bandwidth of the kernel. Therefore, the beta kernel estimator is an adaptive density estimator, as illustrated in Figure \ref{fig:kernel_shape_plot}.
    \begin{figure}[h!]
        \centering
        \includegraphics[width=\textwidth]{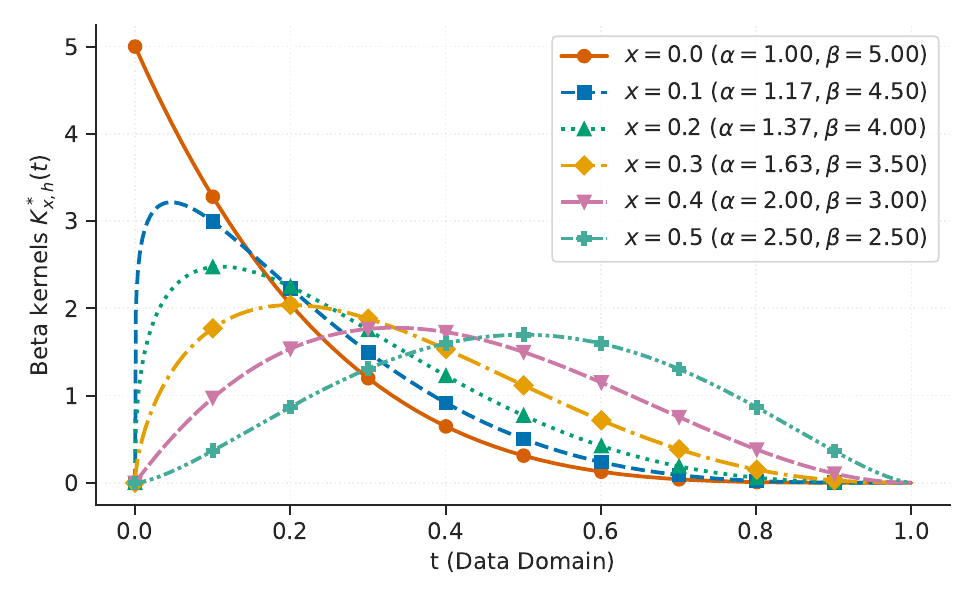}
        \caption{beta kernels $K^*_{x,\bw}(t)$ for bandwidth $h=0.2$.}
        \label{fig:kernel_shape_plot}
    \end{figure}

    \citet{chen1999beta} compared the beta kernel estimator favorably with the local linear estimator \citep{jones1993simple} and the non-negative estimator \citep{jones1996aSimple}, and concluded that the estimator $\ftwo$ was a serious competitor with existing density estimator. 
    \citet{bouezmarni2003consistency} built upon \citeauthor{chen1999beta}'s work by providing a rigorous analysis of the beta kernel estimator, establishing the exact asymptotic behavior of its expected $L_1$ error and proving its uniform weak consistency for continuous densities on a compact support, thereby solidifying its suitability for bounded data by demonstrating its favorable theoretical properties.
   
    It is important to clarify that although the beta kernel is locally adaptive (as its shape changes with $x$), its overall performance is still governed by a single global bandwidth parameter $h$. The positive results of \citet{chen1999beta} hinge on the appropriate choice of $h$. Despite its competitive performance and attractive properties, the beta kernel has not gained traction among practitioners because of its complexity. This is largely because, unlike the ubiquitous Gaussian kernel, it lacks a simple and well-known bandwidth selection rule.
        
    The popularity of the Gaussian kernel is due, in no small part, to \emph{Silverman's rule of thumb} \citep{silverman2018density}, which provides an easy data-driven starting point. An analogous rule of thumb is required to make the beta kernel estimator accessible and to unlock its practical potential. The following section derives a practical bandwidth selection rule.    

\section{A rule of thumb bandwidth estimator}
    The main issue for practitioners seeking to use the beta kernel $\ftwo$ is that the optimal bandwidth \eqref{eq:MISEbandwidth} depends on the true unknown distribution; even if it is known, \eqref{eq:MISEbandwidth} is a complicated expression involving nontrivial integrals that may be difficult to solve. Therefore, we derive a simple ``rule of thumb,” inspired by the well-known Silverman's rule of thumb for the Gaussian kernel. This idea is very simple; instead of considering all possible densities on $[0,1]$, we choose a representative parametric family for which the optimal bandwidth can be computed. For a sample from an unknown distribution, we act as if the density belongs to our representative family, estimate the parameters, and use the bandwidth computed using Eq. \eqref{eq:MISEbandwidth}. If the true distribution is roughly well approximated by our family, the resulting rule-of-thumb bandwidth can be expected to be nearly optimal.

    First, we select a reference parametric family for the unit interval. The beta distribution family is a natural choice for this purpose. It has the correct support and is sufficiently flexible for modeling several data shapes. It is also a standard choice in Bayesian statistics (e.g., as a conjugate prior for the binomial distribution). The density function of a beta random variable with parameters $a, b > 0$ is
    \begin{equation}
        \label{eq:betaDensity}
        f(x) = \frac{x^{a-1}(1-x)^{b-1}}{B(a, b)},
    \end{equation}
    where 
    \begin{equation}
        B(a,b) = \frac{\Gamma(a)\Gamma(b)}{\Gamma(a+b)} = \int_0^1t^{a-1}(1-t)^{b-1}dt
    \end{equation}
    is the beta function, which acts as a normalizing constant, and $\Gamma$ is the gamma function. Importantly, for real numbers $a,b$, the beta function is undefined for $a,b\leq0$.

    The MISE optimal bandwidth \eqref{eq:MISEbandwidth}, assuming that $f$ is a beta density function, can be written as
    \begin{equation}
        \label{eq:MISEbandwidthBeta}
        \btwoopt = \bigg(\frac{1}{2n\sqrt{\pi}}\frac{I_1}{I_2}\bigg)^{2/5},
    \end{equation}
    where
    \begin{equation}
        \label{eq:I1}
        I_1 := \int_0^1\frac{f(x)}{\sqrt{x(1-x)}}dx
    \end{equation}
    and
    \begin{equation}
        \label{eq:I2}
        I_2 := \int_0^1(x(1-x)f''(x))^2dx.
    \end{equation}

    We can now compute
    \begin{equation}\label{eq:I1Solved}
        \begin{aligned}
            I_1 &= \int_0^1\frac{x^{a-1}(1-x)^{b-1}}{B(a,b)\sqrt{x(1-x)}}dx
            = \frac{1}{B(a, b)}\int_0^1\frac{x^{a-1}}{x^{1/2}}\cdot \frac{(1-x)^{b-1}}{(1-x)^{1/2}}dx \\
            &= \frac{1}{B(a,b)}\int_0^1 x^{(a-\frac{1}{2})-1} (1-x)^{(b-\frac{1}{2})-1}dx 
            = \frac{B(a-\frac{1}{2}, b-\frac{1}{2})}{B(a,b)} \\
            &= \frac{\Gamma(a+b)\Gamma(a-\frac{1}{2})\Gamma(b-\frac{1}{2})}{\Gamma(a)\Gamma(b)\Gamma(a+b-1)}
            = \frac{(a+b-1)\Gamma(a-\frac{1}{2})\Gamma(b-\frac{1}{2})}{\Gamma(a)\Gamma(b)},
        \end{aligned}
    \end{equation}
    where the last equality follows from the identity $\Gamma(z+1) = z\Gamma(z)$.
    This formula is valid when $a,b>1/2$. Otherwise, the beta function is undefined. 
    
    The second integral in \eqref{eq:I2} is more complex. This requires substituting the second derivative of the beta density function, which results in a complex integral of a polynomial multiplied by the beta densities. The full derivation is presented in Appendix \ref{app:computingI2}. The complex algebraic manipulations can be handled using a computer algebra system, which yields the final form
    \begin{equation}\label{Eq.{I2Solved}}
        I_2 = \frac{(a-1) (b-1) (a (3 b-4)-4 b+6) \Gamma (2 a-3) \Gamma (2 b-3) \Gamma (a+b)^2}{(2 a+2 b-5) (2 a+2 b-3) \Gamma (a)^2 \Gamma (b)^2 \Gamma (2 a+2 b-6)},
    \end{equation}
    which is valid when $a,b>3/2$.
    
    We can now compute the MISE optimal bandwidth \eqref{eq:MISEbandwidthBeta} by substituting the integral values $I_1$ and $I_2$. The MISE optimal bandwidth for a beta distribution with parameters $a,b>3/2$ is given by
    \begin{equation}\label{eq:MISEoptBeta}
        \btwoopt = \bigg(\tfrac{1}{2n\sqrt{\pi}}\tfrac{(a+b-1) (2 a+2 b-5) (2 a+2 b-3) \Gamma \left(a-\frac{1}{2}\right) \Gamma (a) \Gamma \left(b-\frac{1}{2}\right) \Gamma (b) \Gamma (2 a+2 b-6)}{(a-1) (b-1) (a (3 b-4)-4b+6) \Gamma (2 a-3) \Gamma (2 b-3) \Gamma (a+b)^2}\bigg)^{2/5}.
    \end{equation}
    The expression \eqref{eq:MISEoptBeta} can (with patience or using a computer algebra system) be simplified to
    \begin{equation}
        \btwoopt = \left(\frac{\sqrt{\pi}}{n}\tfrac{(2 a-3) (2 b-3) 2^{-2 a-2 b+5} (2 a+2 b-5) (2 a+2 b-3) \Gamma (2 (a+b-3))}{a (3 b-4)-4 b+6) \Gamma (a+b-1) \Gamma (a+b)}\right)^{2/5},
    \end{equation}
    This reduces the number of evaluations of the gamma function from eight to three, which can be computed more rapidly. 
    A rule of thumb bandwidth estimator which approximates the MISE optimal bandwidth if the data are well-approximated by a beta distribution with parameters $a,b>3/2$ is therefore given by
    \begin{equation}
        \label{eq:MISERule}
        \betareference = \left(\frac{\sqrt{\pi}}{n}\tfrac{(2 \ahat-3) (2 \bhat-3) 2^{-2 \ahat-2 \bhat+5} (2 \ahat+2 \bhat-5) (2 \ahat+2 \bhat-3) \Gamma (2 (\ahat+\bhat-3))}{\ahat (3 \bhat-4)-4 \bhat+6) \Gamma (\ahat+\bhat-1) \Gamma (\ahat+\bhat)}\right)^{2/5},
    \end{equation}
    where $(\ahat, \bhat)$ is estimated from the data, for example, using maximum likelihood estimation (MLE) or the method of moments (MoM). This differs from the ``plain'' rule of thumb used in comparable studies (e.g., \citet{hirukawa2010nonparametric}), which applies a generic Gaussian-style scaling $\hat{\sigma}n^{-2/5}$. Our derivation \eqref{eq:MISERule} incorporates the specific curvature properties of the Beta reference distribution into the constant factor, providing a tighter approximation
    
    Note that $\betareference$ approximates the integral only if $\ahat,\bhat>3/2$. In practice, to prevent floating-point overflow, the calculation can be performed in log space using, for example, the \texttt{gammaln} function provided by the \texttt{scipy} Python package \citep{scipy}. 

    \subsection{Parameter estimation and applicability}
        Our rule of thumb, \eqref{eq:MISERule}, requires an estimate of the parameters of the beta reference distribution. To align with the purpose of a rule of thumb (a fast, ``good enough'' solution in many cases), we recommend the method of moments, as it is easy to compute from the sample mean and variance. Let 
        \begin{equation}
            \overline{x} := \frac{1}{n}\sum_{i=1}^nx_i
        \end{equation}
        be the sample mean, and
        \begin{equation}
            \vbar := \frac{1}{n-1}\sum_{i=1}^n(x_i - \overline{x})^2
        \end{equation}
        be the sample variance. If $\vbar<\xbar(1-\overline{x})$, the MoM estimate of the parameters $a,b$ are
        \begin{equation}
            \begin{aligned}
                \ahat &:= \xbar\bigg(\frac{\xbar(1-\xbar)}{\vbar}-1\bigg) \\
                \bhat &:= (1-\xbar)\bigg(\frac{\xbar(1-\xbar)}{\vbar}-1\bigg).
            \end{aligned}
        \end{equation}
        These estimates are fast and simple to compute and can be plugged into our rule of thumb \eqref{eq:MISERule} to quickly compute $\betareference$.

        Note that the denominator integral $I_2$, and thus our final rule of thumb $\betareference$, is only defined for beta distributions, where $a, b > 3/2$. This constraint excludes U-shaped ($a, b < 1$) and J-shaped distributions. If the parameters $(\ahat, \bhat)$ estimated from the data do not satisfy this constraint, the rule of thumb cannot be applied directly. 

        We have identified the firm constraints that define the domain of applicability of our rule.
        \begin{enumerate}
            \item 
            \emph{MoM constraint:} $\vbar<\xbar(1-\overline{x})$
            \item 
            \emph{Integral constraint:} $\ahat, \bhat > 3/2$. 
        \end{enumerate}
        If these conditions are not met, the rule is not applicable, and a more general bandwidth selector such as cross-validation should be used instead.
        
        This is not a unique flaw in the proposed rule. All plug-in rule-of-thumb methods, including the classic Silverman rule \citep{silverman2018density}, are based on a reference distribution (e.g., Gaussian). Similarly, when the true data-generating process is multimodal, Silverman's rule fails to provide a useful bandwidth. Our constraints clarify these failures.

        It is important to emphasize that the integral constraint ($\ahat, \bhat > 3/2$) is not unique to our approximation but is inherent to the MISE framework itself. For distributions that violate this condition (e.g., U-shaped densities), the roughness functional \eqref{eq:I2} diverges from the original distribution. Consequently, a finite MISE-optimal bandwidth does not strictly exist in the standard sense, rendering heuristic approaches not only convenient but also theoretically necessary.

        Previous attempts to derive analytical bandwidths have avoided this divergence by introducing weights \citep{hirukawa2010nonparametric}. Although mathematically convenient for establishing asymptotic properties, such weighting effectively ignores the fit at the endpoints, which is the region where the beta kernel offers the most value. We evaluate the unweighted MISE directly and propose a specific heuristic for divergent cases.
        
    \subsection{The fallback rule}
        We identified the domain of applicability of $\betareference$. However, in practice, it may be desirable to extend this domain by employing a heuristic rule of thumb when $(\ahat,\bhat)$ falls outside its domain. \citeauthor{chen1999beta}'s analysis showed that that the MISE optimal bandwidth is $\mathcal{O}(n^{-2/5})$. 

        To ensure theoretical transparency, it is necessary to distinguish the two modes of divergence that occur for small shape parameters. First, the unweighted AMISE framework relies on the square integrability of the density’s second derivative, $\int (f''(x))^2 dx$, which diverges when $a, b \leq 1.5$. In the regime of $0.5 < a, b \leq 1.5$, the exact finite-sample MISE remains finite, but the asymptotic Taylor approximation used to derive the bandwidth fails. Second, for extreme boundary accumulation, where $a, b \leq 0.5$ (such as strict U-shaped distributions), the true density itself ceases to be square-integrable, causing the exact finite-sample MISE to diverge. Therefore, the proposed fallback heuristic is essential to provide a stable, well-defined bandwidth across both the regime where the asymptotic approximation breaks down and the regime where the $L_2$ error metric becomes theoretically unbounded.

        We propose a principled closed-form heuristic that uses the estimated parameters $\ahat,\bhat$. Our solution is to define a heuristic scaling factor $C(\ahat,\bhat)$ so that the final bandwidth is 
        $$
            \hheur = C(\ahat,\bhat)n^{-2/5}.
        $$
        By isolating the data-driven scaling factor $C(\hat{a}, \hat{b})$, this formulation guarantees that the fallback bandwidth strictly preserves the $\mathcal{O}(n^{-2/5})$ asymptotic decay rate, which is necessary for balancing variance and boundary bias in the beta kernel estimator, as proven by \citet{chen1999beta}.
        The scaling factor is calculated from the properties of the best-fitting beta distribution, which is defined by $(\ahat,\bhat)$. 
        The heuristic scaling factor is
        \begin{equation}\label{eq:heuristicScaling}
            C(\ahat,\bhat):= \frac{\sqrt{\text{Var}(\ahat,\bhat)}}{1 + |\text{Skewness}(\ahat,\bhat)| + |\text{Excess Kurtosis}(\ahat,\bhat)|}.
        \end{equation}
        
        To calculate this, we must compute three standard values.
        \begin{enumerate}
            \item 
            \textbf{Variance:} Provides the fundamental scale for the scaling factor. This ensures that data with a wider, more dispersed shape (higher variance) receive a proportionally larger $C(\ahat,\bhat)$ and, thus, a larger final bandwidth $h$.
            $$\text{Var}(\ahat,\bhat) = \frac{\ahat\bhat}{(\ahat+\bhat)^2(\ahat+\bhat+1)}$$
            \item 
            \textbf{Skewness:} Serves as a penalty for asymmetry. Highly skewed, ``J-shaped'' distributions require a smaller bandwidth (less smoothing), and the large $|\text{Skewness}|$ term in the denominator correctly shrinks the scaling factor $C(\hat{a},\hat{b})$ to provide this.
            $$\text{Skewness}(\ahat,\bhat) = \frac{2(\bhat-\ahat)\sqrt{\ahat+\bhat+1}}{(\ahat+\bhat+2)\sqrt{\ahat\bhat}}$$
            \item 
        \textbf{Excess Kurtosis:} Serves as a data-driven ``complexity penalty'' that adaptively shrinks the bandwidth for ``spiky'' J-shaped or U-shaped distributions, which require less smoothing.
        $$\text{Excess Kurtosis}(\ahat,\bhat) = \frac{6((\ahat+\bhat)^2(\ahat+\bhat+1) - \ahat\bhat(\ahat+\bhat+2))}{\ahat\bhat(\ahat+\bhat+2)(\ahat+\bhat+3)}$$
        \end{enumerate}

        Importantly, we do not claim that this heuristic is optimal. In the absence of a convergent analytical solution for these divergent cases, the functional form of $C(\ahat, \bhat)$ was constructed to prioritize parsimony and robustness over precision. The numerator, $\sqrt{Var(\hat{a},\hat{b})}$, establishes the fundamental scale, ensuring that the bandwidth remains proportional to the data dispersion. The denominator serves as a robust regularization term. High skewness and excess kurtosis typically indicate a probability mass that concentrates sharply against boundaries (as in J-shaped distributions) or significant deviations from the assumption of a unimodal beta distribution (as in bimodal or U-shaped mixtures). In these ``hard'' regimes, the standard asymptotic approximation fails. Therefore, we employ an unweighted sum of the absolute skewness and excess kurtosis to dampen the bandwidth. This choice is deliberate: by avoiding fitted coefficients, we prevent overfitting to specific test distributions while ensuring that the bandwidth is adaptively reduced as the shape complexity increases. 

        To motivate the specific form, an ablation study is presented in Appendix \ref{app:ablation}. The results demonstrate that omitting any of these higher-order moments renders the estimator vulnerable to catastrophic structural failure.

        \begin{algorithm}[h!]
        \caption{The Rule of Thumb Bandwidth Selection Algorithm}\label{alg:adaptiveRuleOfThumb}
        \begin{algorithmic}
        \Require Data $X = \{x_1, \dots, x_n\}$
        \State $n \gets |X|$
        \State Estimate parameters $(\hat{a},\hat{b})$ from $X$ using the Method of Moments.
        \If{$\hat{a} > 3/2$ \textbf{and} $\hat{b} > 3/2$}
            \Comment{Domain is valid: use the main plug-in rule}
            \State $h \gets \betareference$ defined in Eq. \eqref{eq:MISERule}.
        \Else
            \Comment{Domain is invalid: use the fallback heuristic}
            \State $h \gets \hheur = C(\hat{a},\hat{b}) n^{-2/5}$, where $C(\hat{a},\hat{b})$ is defined in Eq. \eqref{eq:heuristicScaling}.
        \EndIf
        \State \Return $h$
        \end{algorithmic}
        \end{algorithm}

\section{Empirical evaluation: Experimental setup}\label{sec:experimental-setup}
    We designed two comprehensive experiments to assess the performance of the proposed MISE rule-of-thumb bandwidth selector.
    \begin{enumerate}
        \item 
        A large-scale Monte Carlo simulation (Experiment 1) was conducted to compare the performance against a known ground truth ( MISE-optimal bandwidth) across various distributions and sample sizes.
        \item 
        A real-world application (Experiment 2) was used to evaluate the practical performance, scalability, and predictive power of complex and messy datasets using a rigorous cross-validation framework.
    \end{enumerate}
    All experiments were conducted on a Linux server equipped with an Intel Xeon Gold 6526Y CPU and 512 GB of RAM. 
    The computation times report the average wall clock time per fit, measured using a sequential execution model.
    \subsection{Competing methods}\label{sec:compeatingMethods}
        We evaluated a total of 10 methods in our experiments, which are described below.
        \begin{itemize}
            \item 
            \textbf{Proposed method (\rot):} Our fast, analytic bandwidth rule for the beta kernel, based on the Asymptotic MISE (AMISE) formula, with a robust fallback heuristic.
            \item 
            \textbf{Primary ``Gold Standard'' Competitor (\blscv):} The full, numerical LSCV optimization for the beta kernel. This is a slow but established method that \rot was designed to replace.
            \item 
            \textbf{Alternative Kernel Methods:} We compare against two common Gaussian-kernel-based approaches for $[0, 1]$ data:
            \begin{itemize}
                \item 
                \textbf{Logit Transform:} \lsilv (Silverman's rule) and \llscv (LSCV-optimization on logit-transformed data. We selected logit transformation as the baseline. While \citet{geenens2014probit} argue for the theoretical advantages of the probit transform, the logit function remains the canonical mapping for unit-interval data in data science (e.g., as the inverse of the sigmoid activation in machine learning). It represents the standard ``transformation-based'' workflow for practitioners applying Gaussian KDE to bounded data.)
                \item 
                \textbf{Reflection:} \rsilv (Silverman's rule) and \rlscv (LSCV optimization on reflected data).
            \end{itemize}
            \item 
            \textbf{Theoretical Ground Truth (Simulation Only):}
            \begin{itemize}
                \item 
                \bise, \lise, \rise: The bandwidth is calculated by direct minimization of the (unknowable in practice) Integrated Squared Error (ISE). 
                \item 
                \oracle: The theoretical MISE-optimal bandwidth \eqref{eq:MISEbandwidth} is calculated using the true distribution.
            \end{itemize}
        \end{itemize}

    \subsection{Evaluation metrics}
        We evaluated the method based on the following four criteria.
        \begin{itemize}
            \item 
            \textbf{Computation time (s):} The wall-clock time it takes to fit the method on the given data. It measures practical feasibility and scalability.
            \item 
            \textbf{Integrated Squared Error (ISE):} The ground truth $L^2$ error.
            \begin{equation}
                \label{eq:ISEscore}
                ISE(h) := \int_0^1(\hat{f}_h(x) - f(x))^2dx
            \end{equation}
            This was our primary metric in Experiment 1, in which the true density $f$ was known. This could not be computed for the real-world data used in Experiment 2. Lower values indicate better performance.
            \item 
            \textbf{LSCV:} The universal LSCV score, calculated on the full dataset
            \begin{equation}
                \label{eq:LSCVscore}
                LSCV(h) := \int_0^1 \hat{f}_h(x)^2dx - \frac{2}{n}\sum_{i=1}^n\hat{f}_{h,(-i)}(x_i),
            \end{equation}
            where $\hat{f}_{h,(-i)}(x_i)$ is the leave-one-out (LOO) density estimate at $x_i$, fitted to the full dataset, except $x_i$. For computational efficiency in the large-scale Monte Carlo simulation (Experiment 1), this score was computed using 10-fold cross-validation to approximate the expensive LOO term. For real-world applications (Experiment 2), the theoretically exact LOO term was computed directly to ensure maximum accuracy. This is our primary performance metric for the ``hard'' distributions in Experiment 1 (where the ISE is not computable) and a key summary metric in Experiment 2. Lower is better.
            \item 
            \textbf{Per-Fold Statistical Metrics (For Experiment 2):} To achieve statistical power in our real-world analysis, we ran a separate 10-fold cross-validation procedure. To achieve sufficient statistical power and avoid reliance on a single data shuffle, we employed a 10-repetition, 10-fold cross-validation procedure. This generated 100 scores for each method, which were compared using the robust nonparametric Wilcoxon signed-rank test \citep{wilcoxon1945test}.
            \begin{itemize}
                \item
                \textbf{Mean Held-out Density:} We computed the mean held-out density for each of the 100 folds. This metric represents the data-driven cross-validation term ($\frac{1}{n}\sum \hat{f}_{h,(-i)}(x_i)$) of the LSCV objective. Because this term is subtracted in the LSCV formula, higher values indicate more accurate density estimates.
            \end{itemize}
        \end{itemize}

    \subsection{Experiment 1: Monte Carlo simulation}
        Objective: This study aimed to evaluate the performance of \rot in a controlled environment. We measured its speed, robustness, and (most importantly) its true accuracy against the known ``ground truth'' optimal bandwidth, as well as other competing methods. 

        We used eight distributions, which can be categorized as follows:
        \begin{itemize}
            \item 
            \textbf{``Nice'' (Bell-shaped but possibly skewed)}. These are $B(5, 5))$, $B(2, 12)$, $\NT(0.5, 0.15)$, $\NT(0.7, 0.15)$. Here, $B$ denotes the beta distribution that precisely satisfies the assumptions under which \rot was derived and $\NT$ denotes the truncated Gaussian (TG) distribution. All ``nice'' distributions are expected to work well with \rot.
            \item 
            \textbf{``Hard'' (U-shape, J-shape and boundary case):} These distributions are $B(0.5, 0.5)$, which is U-shaped, $B(0.8, 2.5)$, which is J-shaped. These fall outside the domain of applicability for \rot, which allows testing of the fallback rule. We also included the boundary case $B(1.5, 1.5)$, whose parameters lie exactly on the boundary of the domain of applicability in the analysis.
            \item 
            \textbf{``Tricky'' (Bimodal):} We included a mixture of $B(10, 30)$ and $B(30,10)$ with mixing parameter 1/2, which is a bimodal distribution. This is particularly challenging for any rule of thumb.
        \end{itemize}

        For each distribution and method listed in Section \ref{sec:compeatingMethods}, we fit the method to samples of sizes $n=50, 100, 250, 500, 1000, 2000$. However, for the ``hard'' distributions, the \oracle methods and the methods that directly minimize the ISE are not available because the integrals either diverge or are numerically unstable. We ran 1000 independent trials for each distribution, method, and sample size, resulting in 48000 repetitions. 

        For each method, distribution, and sample size, we recorded the computation time, LSCV score, and, with the exception of the ``hard'' distributions (where it cannot be computed), we also recorded the ISE score. For the \rot, we also recorded whether the heuristic fallback rule was applied. To assess the statistical significance of the performance differences across the 1,000 trials, we applied the nonparametric Wilcoxon signed-rank test.

    \subsection{Experiment 2: Real-world application}\label{subsec:real-world-experiment}
        The goal of this experiment was to transition from controlled simulations to real-world scenarios. We assessed the practical performance of the bandwidth selectors on complex, ``messy'' data, in which the true distribution was unknown. This experiment was designed to provide statistically powerful performance comparisons.

        We used the ``Communities and Crime'' dataset, which is publicly available through the UCI machine learning repository \citep{communities_and_crime_183}. These data are naturally bounded in $[0,1]$, making them suitable for our purpose. Specifically, we used three variables.
        \begin{itemize}
            \item 
            PctKids2Par (percentage of kids in two-parent households)
            \item 
            PctPopUnderPov (percentage of population under poverty)
            \item 
            PctVacantBoarded (percentage of vacant housing that is boarded up)
        \end{itemize}

        We compared the six aforementioned ``practical methods'' listed in Section \ref{sec:compeatingMethods}. The \oracle and methods that directly minimize the ISE were excluded because they assume knowledge of the true density.

        We performed 10 repetitions of 10-fold cross-validation to compare the methods. We then used the Wilcoxon signed-rank test to assess significance, as it is a robust nonparametric test well-suited for this comparison.

\section{Experimental results}
    We conducted two experiments, as described in Section \ref{sec:experimental-setup}. The results from both the large-scale simulation and real-world data application provide a comprehensive and consistent narrative: the proposed \rot (MISE rule-of-thumb) is not only a feasible and scalable alternative to \blscv (LSCV optimization), but also a more accurate and stable estimator of the MISE-optimal bandwidth.

    \subsection{Experiment 1: Monte Carlo simulation results}
        The simulation (1000 trials per configuration) was designed to test scalability, ground-truth accuracy (ISE), with the exception of the ``hard distributions, '' where this metric is unavailable, and robustness. 

        The LSCV scores are summarized in Table \ref{tab:main_table_lscv}, where we report the LSCV score of each method, averaged over the ``nice,” ``bimodal,” and ``hard'' distributions and all sample sizes. The results reveal a clear robustness–trade-off. 
        On ``nice'' distributions, the computationally expensive \blscv method achieves a mean score of -2.2716 (median: -1.9155), which is a statistically significant, albeit small, improvement over our \rot's mean of -2.2667 (median: -1.9134).

        However, this advantage is reversed when the data become complex. On the 'hard' distributions, the \rot was statistically significantly better ($p < 0.001$) than its slow \blscv counterpart and the \rlscv estimator. 
        On the 'bimodal' distribution, our rule significantly outperformed \blscv, although the \rlscv estimator achieved a better LSCV score (it is worth noting that the bimodal distribution satisfies the zero-derivative boundary condition that is enforced by the reflection method).
        This strongly suggests that the LSCV optimization process becomes unstable and fails to find a good bandwidth for these data types, whereas our rule remains robust. Notably, while \lsilv achieves the best mean score (-3.7532) on ``hard'' data, the median score (-1.1105) is worse than the median score of all non-logit methods. An analysis of the simulation data reveals that the mean score of the \lsilv is heavily skewed by massive negative outliers, likely caused by the logit function mapping data near the boundary to $\pm\infty$. When comparing the performance using the Wilcoxon signed-rank test, \rot outperformed \lsilv ($p<0.001$).

        \begin{table}[h!]
            \centering
            \resizebox{\textwidth}{!}{
            \input{tables/main_table_lscv}
            }
            \caption{Mean LSCV scores (median in parentheses) across distribution groups. Bold indicates the best median per group. Significance of Wilcoxon signed-rank tests vs. the reference method: $^{*}p<0.05$, $^{**}p<0.01$, $^{***}p<0.001$.}
            \label{tab:main_table_lscv}
        \end{table}

        These findings are visually confirmed in Figure \ref{fig:LSCV_Score_vs_N}. This figure plots the mean LSCV score (lower is better) against the sample size ($n$) for all eight test distributions. Across all panels, the performance of the proposed rule, \rot (solid line), is highly competitive. On the ``nice'' and ``bimodal'' distributions, it closely tracks the performance of the best (but slow) LSCV-based methods. For the ``hard'' distributions, such as $B(0.5, 0.5)$ and $B(0.8, 2.5)$,  stability is evident, as it maintains consistently low scores.
        
        \begin{figure}[h!]
            \centering
            \includegraphics[width=\textwidth]{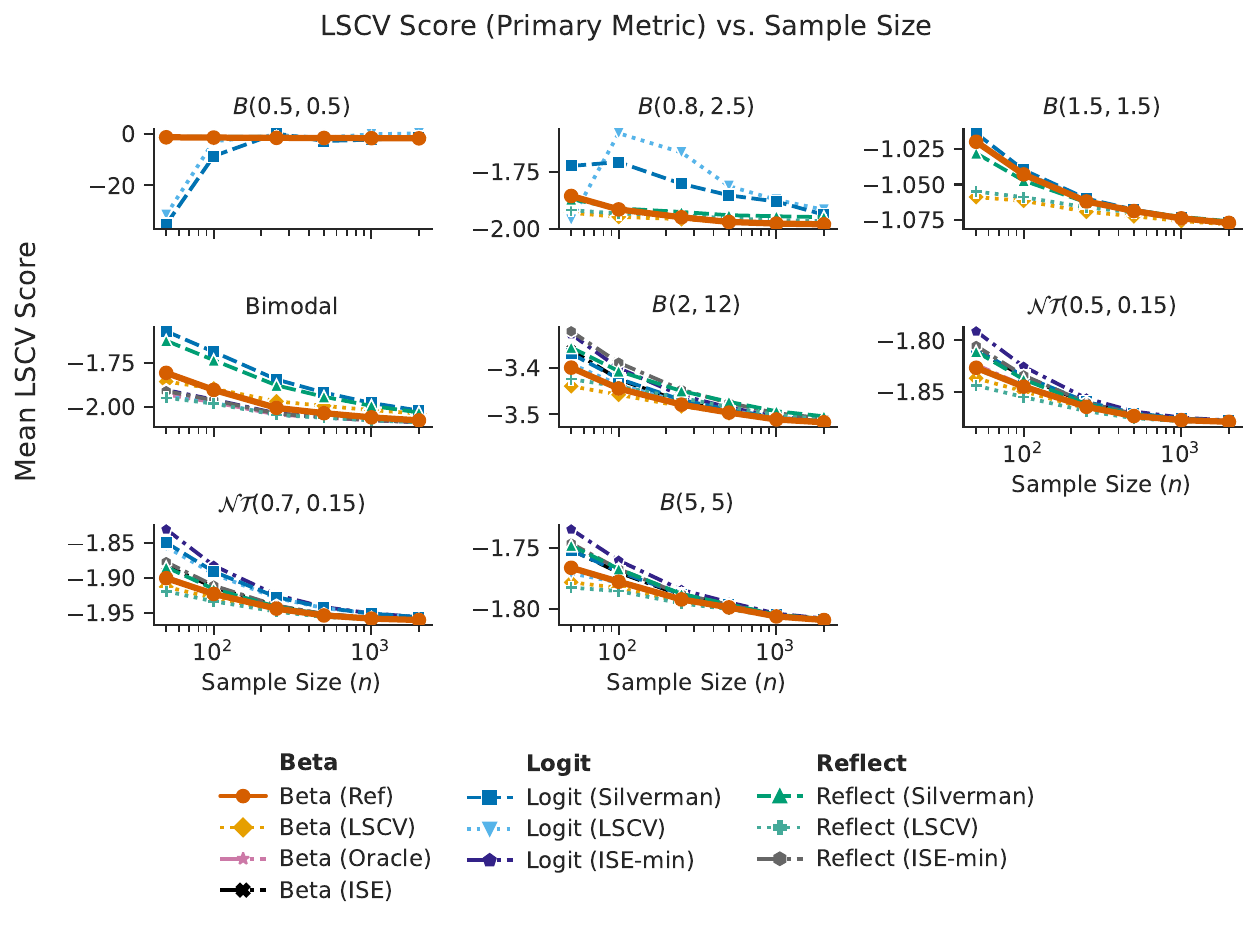}
            \caption{Mean LSCV score as a function of sample size ($n$) across all eight test distributions. Our proposed method, \rot (solid line), is highly competitive in this regard. It closely tracks the performance of the best slow-optimization methods on the ``nice'' distributions (e.g., $B(2,12)$) while demonstrating superior performance and stability on the ``hard'' (e.g., $B(0.5, 0.5)$) and ``bimodal'' distributions as $n$ increases.}
            \label{fig:LSCV_Score_vs_N}
        \end{figure}

        To validate our findings, we compared the mean integrated squared error (ISE) in Table \ref{tab:main_table_ise}. The ISE scores, which are omitted for ``hard'' distributions owing to their instability, confirm and strengthen our findings from the LSCV analysis.

        On ``nice'' distributions, our \rot (mean ISE: 0.0313; median: 0.0149) is statistically significantly more accurate ($p < 0.001$) than all other practical methods, including the slow \blscv (mean ISE: 0.0358; median: 0.0161) and the competing fast rules \lsilv  (mean ISE: 0.0431; median: 0.0206) and \rsilv (mean ISE: 0.0413; median: 0.0219).
        
        The results of the ``bimodal'' distribution are even more stark. The fast Gaussian-based rules (\lsilv, \rsilv) failed completely, producing extremely high mean error scores (0.2377 and 0.2067, respectively) in this study. In contrast, \rot remains highly robust (0.0914), again proving to be statistically superior ($p < 0.001$) to its slow counterpart, \blscv (0.1104). 
        Crucially, this robustness is not because the beta reference distribution successfully models bimodality, but rather because the method successfully detects that it cannot. For bimodal mixtures, the method-of-moments estimates $(\hat{a}, \hat{b})$ typically fall into the invalid domain ($\leq 1.5$), effectively flagging a violation of the unimodal assumption. This automatically triggers the fallback heuristic (as confirmed in the detailed simulation results (see Supplementary Material), where the fallback rate is $>99\%$). Unlike standard Gaussian reference rules, which blindly apply a global bandwidth that over-smooths multimodal structures, our approach defaults to a conservative shape-penalized bandwidth that preserves the density features.
        
        \begin{table}[h!]
            \centering
            \input{tables/main_table_ise}
            \caption{Mean ISE scores (median in parentheses) across distribution groups. Bold indicates the best median per group. Significance of Wilcoxon signed-rank tests vs. the reference method: $^{*}p<0.05$, $^{**}p<0.01$, $^{***}p<0.001$.}
            \label{tab:main_table_ise}
        \end{table}

        These results are visually confirmed in Figure \ref{fig:ISE_Score_vs_N}, which shows the ISE with respect to the sample size. The line for \rot (Beta (Ref)) consistently tracks just above the oracle methods (\bise, \oracle) and well below the competing kernel families, demonstrating its superior performance.
        
        \begin{figure}[h!]
            \centering
            \includegraphics[width=\textwidth]{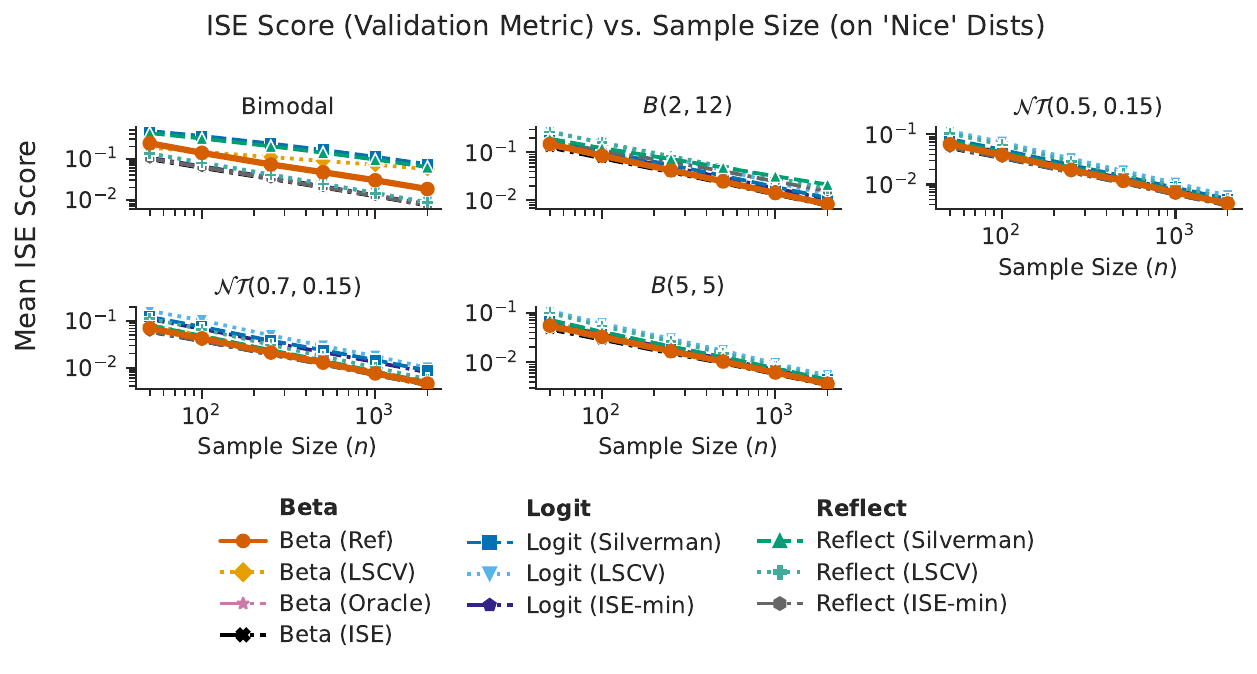}
            \caption{Mean ISE (log-scale) as a function of sample size ($n$, log-scale) for the ``nice'' and ``bimodal'' distributions. This plot visually confirms the findings presented in Table \ref{tab:main_table_ise}. Our proposed rule, \rot (solid line), is shown to be highly accurate, with its performance line consistently tracking just above the oracle methods (\bise, \oracle) and visibly outperforming all competing fast rules (\lsilv , \rsilv) and the slow \blscv.}
            \label{fig:ISE_Score_vs_N}
        \end{figure}

        Finally, we analyze the practical cost of these methods in Table \ref{tab:main_table_time} and Figure \ref{fig:Comp_Time_vs_N}. The results are unambiguous: the \rot, along with the other fast rules (\lsilv, \rsilv), is instantaneous, requiring, on average, 0.0001s regardless of the complexity of the distribution.

        This is in stark contrast to all optimization-based methods. For example, the \blscv method is more than 35, 000 times slower (3.5567s) on ``nice'' data, with \rlscv being more than 217, 000 times slower (21.7347s). Figure \ref{fig:Comp_Time_vs_N} visualizes this difference, showing that the LSCV and ISE methods are orders of magnitude slower than the fast methods, which are clustered on the $x$-axis.
        
        In summary, the \rot provides a highly competitive and practical choice for density estimation in bounded domains, as it uniquely balances computational efficiency with robustness against complex density shapes.

        \begin{table}[h!]
            \centering
            \resizebox{\textwidth}{!}{
            \input{tables/main_table_time}
            }
            \caption{Mean computation times in seconds across distribution groups.}
            \label{tab:main_table_time}
        \end{table}
        
        \begin{figure}[h!]
            \centering
            \includegraphics[width=\textwidth]{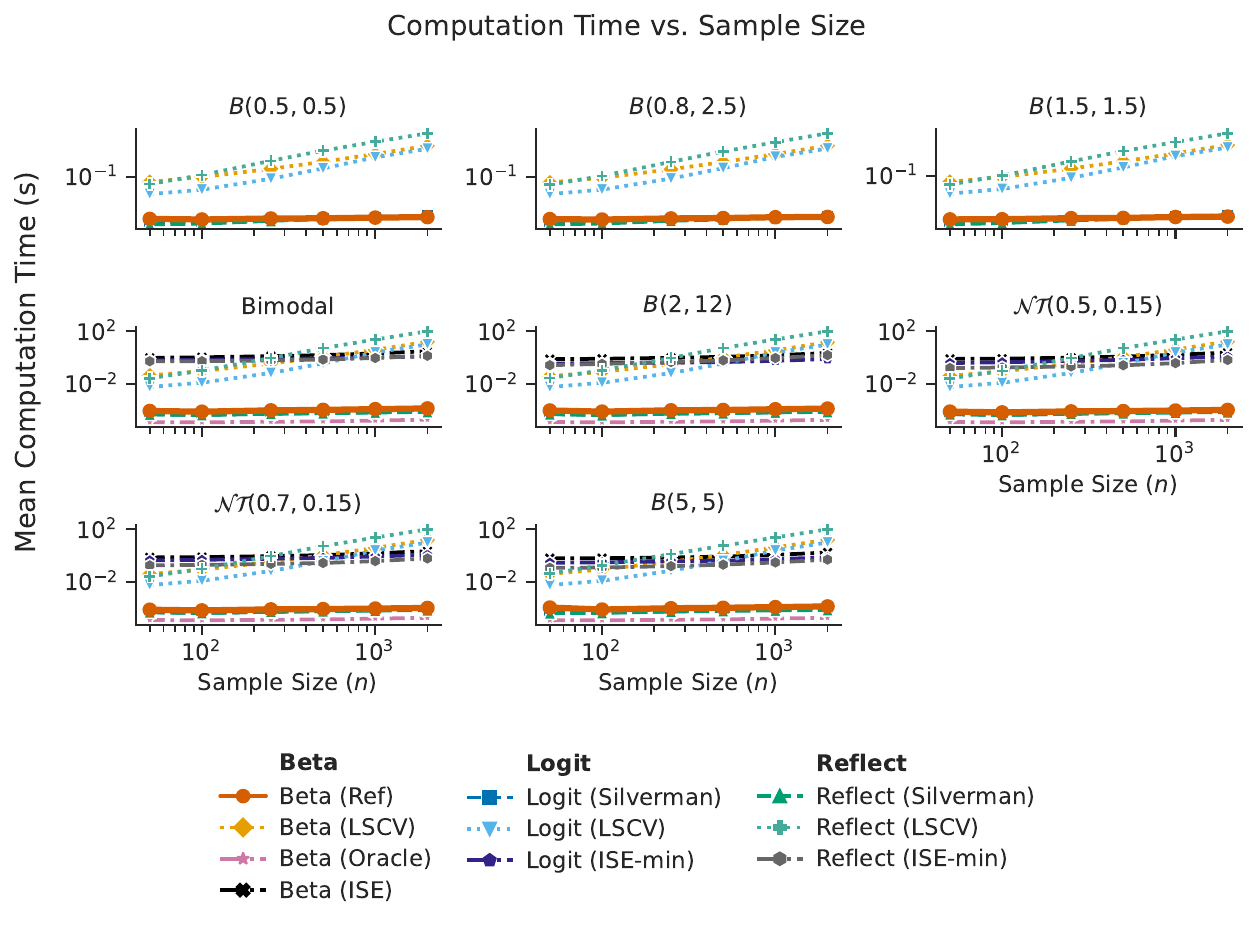}
            \caption{Mean computation time (log-scale) as a function of sample size ($n$, log-scale) across all eight test distributions. This plot visually illustrates the results in Table \ref{tab:main_table_time}. It shows a clear separation between the two performance classes: the fast methods, including our \rot (Beta (Ref)), which are clustered at the bottom with a near-constant cost of approximately $10^{-4}$ s. In contrast, all Slow (LSCV) and Benchmark (Oracle) methods are orders of magnitude slower, and their computational cost clearly increases with the sample size $n$.}
            \label{fig:Comp_Time_vs_N}
        \end{figure}

        To understand why our \rot achieves such a high level of accuracy, Figure \ref{fig:Bandwidth_vs_N_NiceDists} shows the bandwidth ($h$) that was selected. The figure compares the bandwidth from our rule (solid line) to the optimal bandwidths derived from the oracle methods (\bise and \oracle) for the ``nice'' and ``bimodal'' distributions. The plot clearly shows that the bandwidth selected by \rot successfully tracked the true optimal bandwidth across all sample sizes and distributions. This demonstrates that our rule is not only a fast approximation but also effectively identifies and converges to the asymptotically optimal bandwidth for the beta kernel for all datasets.
        
        \begin{figure}[h!]
            \centering
            \includegraphics[width=\textwidth]{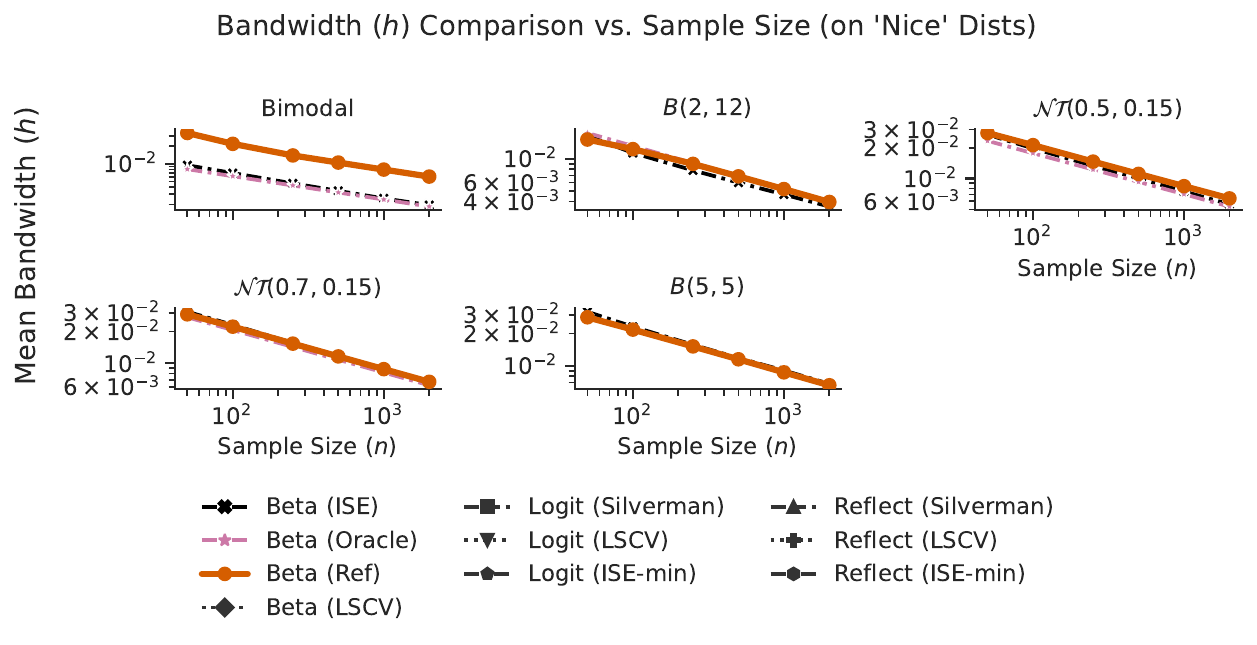}
            \caption{Mean selected bandwidth ($h$, log-scale) as a function of sample size ($n$, log-scale) for the ``nice'' and ``bimodal'' distributions. This plot compares the bandwidth selected by our proposed fast rule (``Beta (Ref)'') to the optimal bandwidths derived from the oracle methods (\bise, ``Beta (ISE-min)'' and \oracle, ``Beta (Oracle)''). The \rot bandwidth is shown to closely track the optimal oracle bandwidths across all distributions and sample sizes, visually confirming the accuracy of its derivation.}
            \label{fig:Bandwidth_vs_N_NiceDists}
        \end{figure}

    \subsection{Experiment 2: Real-world application results}\label{subsec:real-world-results}

        In this experiment, we evaluated the performance of bandwidth selectors on three real-world variables from the Communities and Crime dataset: PctKids2Par, PctPopUnderPov, and PctVacantBoarded. Unlike in the simulation, the actual density was unknown. Therefore, we relied on the LSCV score (computed using the exact leave-one-out formula on the full dataset) as our primary accuracy metric, along with the computation time, to evaluate scalability. To understand the mechanism underlying the performance of the proposed rule, we recorded the percentage of trials triggered by the \rot fallback heuristic.

        The density estimates are shown in Figure \ref{fig:experiment_2_visual_fits} and the quantitative results are summarized in Table \ref{tab:experiment2}. Detailed statistical test results, including Wilcoxon signed-rank p-values and secondary log-likelihood metrics, are provided in the Supplementary Material.
        \begin{figure}[h!]
            \centering
            \includegraphics[width=\textwidth]{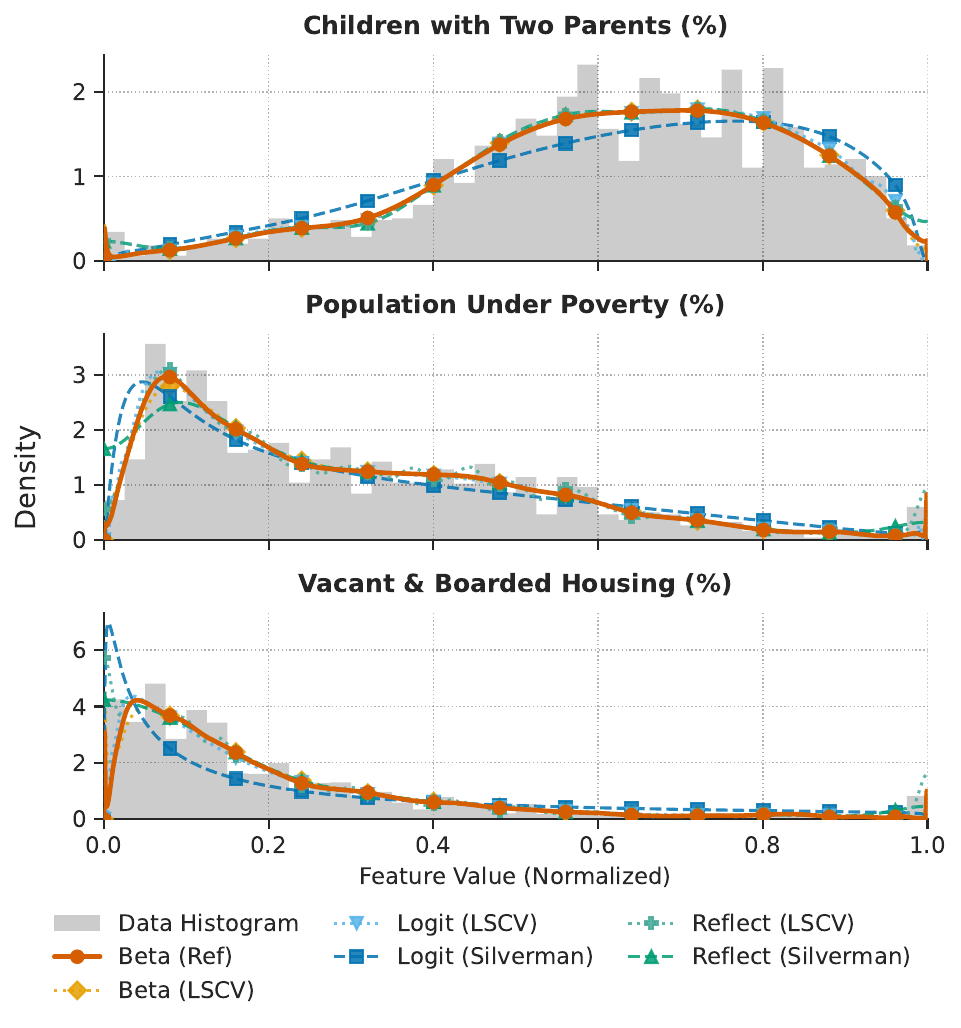}
            \caption{Density estimates for \texttt{PctKids2Par}, \texttt{PctPopUnderPov}, and \texttt{PctVacantBoarded}, comparing the proposed \rot against optimization-based competitors and alternative kernel families.}
            \label{fig:experiment_2_visual_fits}
        \end{figure}

        The results reveal a striking dichotomy driven by the nature of the data. On the ``nice” bell-shaped PctKids2Par distribution (Figure \ref{fig:experiment_2_visual_fits}a), the data were well-approximated by a beta distribution. Consequently, the \rot rarely engages its fallback heuristic (0.0\% usage; see Table \ref{tab:experiment2}). In this regime, the rule performed as a high-quality approximation, achieving an LSCV score of -1.4030, which was virtually identical to the optimization-based \blscv (-1.4031) but was computed approximately 75000 times faster (0.0002s vs. 15.0s).

        The strengths of the proposed rule are most apparent in the ``messy'', boundary-biased distributions: PctPopUnderPov and PctVacantBoarded (Figures \ref{fig:experiment_2_visual_fits}b and c). These distributions violate the standard beta assumptions (e.g., $\ahat, \bhat > 3/2$), causing the standard LSCV optimization to become unstable and producing under-smoothed, ``wiggly'' estimates that overfit the data spikes.

        In contrast, the \rot automatically detected this violation and engaged its fallback heuristic in 100\% of the cases (Table \ref{tab:experiment2}). This mechanism acts as a robust regularizer that produces stable and smooth density estimates. On PctPopUnderPov, this stability translates into a superior LSCV score for \rot (-1.5949) compared to the unstable \blscv (-1.5913). Although the computationally expensive \llscv and \rlscv methods achieved the lowest (best) LSCV scores overall, they required between 7.1 and 88.3 s for computation. The proposed \rot offers a unique value proposition: it provides an instantaneous and robust estimate that avoids the pathological failure modes of standard optimization while outperforming naive Gaussian rules (\lsilv, \rsilv).
        
        \begin{table}[h!]
            \centering
            \resizebox{\textwidth}{!}{
            \input{tables/experiment_2_table}
            }
            \caption{Experiment 2 results on real-world datasets. LSCV scores (lower is better), mean heldout density with median in parentheses (higher is better; bold indicates best median), computation time, and fallback rate. Significance of Wilcoxon signed-rank tests vs. the reference method: $^{*}p<0.05$, $^{**}p<0.01$, $^{***}p<0.001$.}
            \label{tab:experiment2}
        \end{table}

\section{Concluding discussion}
    The beta kernel estimator has long been recognized as a theoretically superior alternative to the Gaussian kernel for bounded-support data. By naturally matching the support of the kernel to the domain of the data, it eliminates boundary bias without artifacts introduced by reflection or transformation. However, despite these clear advantages, this method remains a specialist tool, limited by a single practical bottleneck: the lack of a simple, reliable, and fast bandwidth-selection method. In this study, we aimed to remove these barriers.

    By minimizing the asymptotic mean integrated squared error (AMISE) for a beta reference distribution, we derived a fast analytic rule of thumb for the bandwidth $h$. The empirical results are unambiguous: our method matches the accuracy of computationally expensive LSCV optimization on standard distributions while offering a computational speedup of over five orders of magnitude compared with the latter. As with any rule of thumb based on a single reference distribution, our method is theoretically suboptimal for multimodal densities. However, our simulations indicate that even in these ``bimodal'' scenarios, the proposed fallback rule remains highly effective, outperforming the numerically unstable LSCV in terms of integrated squared error (ISE). This suggests that the stability of parametric approximation often outweighs the theoretical flexibility of optimization methods in finite-sample settings.

    We also explored bandwidth selectors based on minimizing Kullback-Leibler divergence (approximating the integrated chi-squared error). Although this metric simplifies the derivation by canceling density terms, yielding closed-form solutions that are simple polynomials in $\ahat$ and $\bhat$, it implicitly weights the errors by $1/f(x)$. This weighting causes severe integrability issues at the boundaries for $a,b < 2$, effectively prioritizing the tail fit over the mode. This mirrors the convergence challenges noted in the bias-correction literature, which often necessitates aggressive weighting functions (e.g., $w(x)=x^5(1-x)^5$ in \citet{hirukawa2010nonparametric} or $w(x)=x^3(1-x)^3$ in \citet{jones2007kernel}) to remain solvable. Our unweighted $L_2$ approach, while algebraically more complex and involving the roughness functional $I_2$, provides a more balanced global fit and avoids these artificial stabilizers.

    Crucially, we address the practical reality of ``hard'' (U-shaped and J-shaped) distributions, in which standard asymptotic approximations frequently fail. Consequently, our proposed method functions as a composite bandwidth selector: it utilizes the rigorous AMISE-derived formula for standard distributions but automatically transitions to a principled skewness-kurtosis heuristic when the data violate the regularity conditions of the reference family. This hybrid approach ensures robustness, preventing the numerical instability observed in the standard LSCV, while yielding superior density estimates in difficult boundary-concentrated cases.

    We also explicitly addressed the theoretical nuances of probability mass conservation. Although the unnormalized beta kernel estimator does not strictly integrate to unity in finite samples, we prove that the deviation decays linearly with the bandwidth ($\mathcal{O}(h)$) and is negligible in practice (typically $<1\%$). By retaining the unnormalized form, we preserved the natural boundary adaptivity of the estimator and demonstrated its superior performance under the MISE metric, demonstrating that the benefits of bias reduction far outweigh the costs of minor mass deviations. However, for practical deployment, particularly in visualization pipelines or probabilistic modeling, we recommend a simple post-hoc renormalization ($\ftwo^{norm} = \ftwo / \int \ftwo$) to ensure a strict unit probability mass; our provided Python package includes this as a built-in option.

    Ultimately, this study positions the beta kernel as a drop-in replacement for the Gaussian kernel in modern data science workflows. With the accompanying open-source Python package (see Appendix \ref{app:package}), practitioners can now leverage the superior boundary properties of the beta kernel with the same computational ease and $\mathcal{O}(1)$ efficiency as standard methods. Future work could extend this closed-form derivation logic to other asymmetric kernels, such as the gamma or inverse Gaussian kernels, to further democratize boundary-corrected density estimation across different bounded domains.

    We explicitly note two inherent limitations of the proposed method. First, as is the case with any reference-based bandwidth selector (such as Silverman's rule for the Gaussian kernel), the closed-form selector $h_{ref}$ is mathematically sub-optimal for densities that cannot be well-approximated by the chosen reference family, such as strongly bimodal or multimodal mixtures. Second, while many U-shaped and J-shaped densities are natural members of the beta reference family, the rigorous derivation of $h_{ref}$ is only valid when the estimated shape parameters satisfy $\hat{a}, \hat{b} > 1.5$. Below this threshold, the unweighted asymptotic roughness functional diverges, meaning a finite MISE-optimal bandwidth strictly does not exist. Consequently, for these boundary-accumulating shapes, the method must transition to the heuristic fallback rule ($h_{heur}$) to regularize the bandwidth. While our ablation study demonstrates that this hybrid approach remains highly effective in practice, deriving a strictly convergent, closed-form asymptotic approximation for these divergent regimes remains a challenging open problem for future research.
    
\section{Acknowledgment}
    The authors would like to thank the Editor, Associate Editor, and two anonymous reviewers for their highly constructive feedback and insightful suggestions, which have significantly improved the theoretical depth and empirical clarity of this manuscript. 

    During the preparation of this manuscript and the accompanying response to reviewers, the authors utilized AI-assisted technologies. Paperpal (overleaf plug-in version 2.0.3) was employed for grammar checking and language polishing. Gemini 3.1 Pro was used as a conversational sounding board to facilitate conceptual discussions during the writing process. Additionally, GitHub Copilot was used to assist in structuring and formatting the Python code provided in the supplementary materials. Following the use of these tools, the authors rigorously reviewed, edited, and verified all content, and take full responsibility for the manuscript's originality, methodology, and scientific findings.

\section{Funding}
    This work was supported by the Swedish Knowledge Foundation through the SPARK Research Environment at Jönköping University (Project PREMACOP, grant no. 20220187)

\section{Disclosure statement}\label{disclosure-statement}

    The authors declare no conflicts of interest.

\section{Data Availability Statement}\label{data-availability-statement}
The ``Communities and Crime'' dataset \citet{communities_and_crime_183} analyzed in this study is publicly available from the UCI Machine Learning Repository. The scripts required to download the specific subsets used for the cross-validation experiments in Sections \ref{subsec:real-world-experiment} and \ref{subsec:real-world-results} are available at 
\href{https://github.com/egonmedhatten/beta-kernel-reproduce-paper}{https://github.com/egonmedhatten/beta-kernel-reproduce-paper}.

\bibliography{references}

\appendix

\section{Computing $I_2$}\label{app:computingI2}
    We seek to find a closed-form solution to the integral \eqref{eq:I2} (reproduced here for convenience)
    \begin{equation}
        I_2 := \int_0^1(x(1-x)f''(x))^2dx.
    \end{equation}

    Our first order of business is to derive a suitable expression for the second derivative of the beta density. For ease of notation, write
    \begin{equation}\label{eq:betaDensityManipulate}
        f(x) = Cx^{a-1}(1-x)^{b-1}
    \end{equation}
    where $C := \frac{1}{B(a,b)} = \frac{\Gamma(a+b)}{\Gamma(a)\Gamma(b)}$.
    Differentiating twice yields
    \begin{equation}
        \begin{aligned}
            f''(x) = C\bigg((a-2)(a-1)x^{a-3}(1-x)^{b-1} &- 2(a-1)(b-1)x^{a-2}(1-x)^{b-2} \\
            &+ (b-2)(b-1)x^{a-1}(1-x)^{b-3}\bigg).
        \end{aligned}
    \end{equation}
    Next, we factor out the lowest powers of $x$ and $(1-x)$, that is, $x^{a-3}(1-x)^{b-3}$. We get
    \begin{equation}\label{eq:secondDerivativePolyFactor}
        f''(x) = Cx^{a-3}(1-x)^{b-3}P_2(x),
    \end{equation}
    where 
    \begin{equation}\label{eq:integrandPolynomial}
        P_2(x) := (a-2)(a-1)(1-x)^2 - 2(a-1)(b-1)x(1-x) + (b-2)(b-1)x^2
    \end{equation}
    is a polynomial of degree two in $x$. The integrand in \eqref{eq:I2} can thus be written as 
    \begin{equation}\label{eq:MISEintegrand}
        T(x) := (x(1-x)Cx^{a-3}(1-x)^{b-3}P_2(x))^2 = C^2x^{2a-4}(1-x)^{2b-4}P^2_2(x).
    \end{equation}
    The factor $P_2^2(x)$ is a polynomial of degree four, that can be written as $\sum_{k=0}^4d_kx^k$, where $d_k$ are coefficients (rather complicated expressions in $a$ and $b$). 
    Therefore, we can rewrite \eqref{eq:MISEintegrand} as
    \begin{equation}\label{eq:MISEintegrandReady}
        T(x) = C^2 \sum_{k=0}^4d_kx^{2a-4+k}(1-x)^{2b-4}.
    \end{equation}
    Plugging in the integrand \eqref{eq:MISEintegrandReady} into \eqref{eq:I2}, yields
    \begin{equation}
        \begin{aligned}
            I_2 &= \int_0^1C^2 \sum_{k=0}^4d_kx^{2a-4+k}(1-x)^{2b-4}dx \\
            &= C^2 \sum_{k=0}^4d_k\int_0^1x^{2a-4+k}(1-x)^{2b-4}dx \\
            &= C^2 \sum_{k=0}^4d_k B(2a-3+k,2b-3).
        \end{aligned}
    \end{equation}
    This is the sum of five beta functions weighted by the rather complicated coefficients $d_k$. This is defined only if $a,b>3/2$. A computer algebra system (CAS) was used to symbolically expand $P_2^2(x)$, compute the sum, and simplify the resulting expression using gamma function identities. The final result is
    \begin{equation}
        I_2 = \frac{(a-1) (b-1) (a (3 b-4)-4 b+6) \Gamma (2 a-3) \Gamma (2 b-3) \Gamma (a+b)^2}{(2 a+2 b-5) (2 a+2 b-3) \Gamma (a)^2 \Gamma (b)^2 \Gamma (2 a+2 b-6)}
    \end{equation}
    which holds true for $a,b>3/2$.

\section{Absolute mass error}\label{app:AME}
    This appendix provides a rigorous derivation of the asymptotic deviation from the unit probability mass for the beta kernel estimator $\ftwo$. We first derive the general form of the integrated bias in terms of the second derivative of the density (Proposition \ref{prop:massDeviationGeneral}), and subsequently evaluate this integral analytically using integration by parts (Proposition \ref{prop:massDeviationSpecific}).
    
    \begin{proposition}\label{prop:massDeviationGeneral}
    Let $h$ be the smoothing bandwidth, and assume that the true density $f$ is twice continuously differentiable in $[0,1]$. The expected total probability mass of the unnormalized beta kernel estimator $\ftwo$ satisfies
    \begin{equation}
    \int_0^1 \mathbb{E}[\ftwo(x)] dx = 1 + \frac{h}{2} \int_0^1 x(1-x)f''(x) dx + \mathcal{O}(h^2)
    \end{equation}
    \end{proposition}
    
    \begin{proof}
    Let $B(x) = \mathbb{E}[\ftwo(x)] - f(x)$ denote the bias. We split the total integral into boundary regions $[0, 2h) \cup (1-2h, 1]$ and interior regions $[2h, 1-2h]$.
    $$ \int_0^1 B(x) dx = \underbrace{\int_0^{2h} B(x) dx + \int_{1-2h}^1 B(x) dx}_{\text{Boundary Contribution}} + \underbrace{\int_{2h}^{1-2h} B(x) dx}_{\text{Interior Contribution}} $$
    
    \noindent \textbf{1. Boundary Contribution:}
    \cite{chen1999beta} shows that the bias $B(x)$ is of order $\mathcal{O}(h)$ everywhere on $[0,1]$. The combined width of the boundary regions was $4h$. Therefore, the absolute contribution of the boundaries is bounded by
    $$ \left| \int_{\text{Boundaries}} B(x) dx \right| \leq \text{Width} \times \max |B(x)| = 4h \times \mathcal{O}(h) = \mathcal{O}(h^2) $$
    \noindent \textbf{2. Interior Contribution:}
    In the interior region, according to \citet{chen1999beta}, the bias is given by $B(x) = \frac{h}{2}x(1-x)f''(x) + \mathcal{O}(h^2)$. Integrating this term:
    $$ \int_{2h}^{1-2h} B(x) dx = \frac{h}{2} \int_{2h}^{1-2h} x(1-x)f''(x) dx + \mathcal{O}(h^2) $$
    The limits of the integral can be extended from $[2h, 1-2h]$ to the full interval $[0,1]$. The error introduced by adding the boundary segments back into is once again the integral of a bounded function over a region of width $4h$, which is $\mathcal{O}(h^2)$. Thus:
    $$ \int_{2h}^{1-2h} B(x)dx = \frac{h}{2} \int_0^1 x(1-x)f''(x) dx + \mathcal{O}(h^2) $$
    The addition of the true unit mass $\int_0^1 f(x)dx = 1$ completes the proof.
    \end{proof}
    
    \begin{proposition}\label{prop:massDeviationSpecific}
    For any twice continuously differentiable probability density function $f$ on $[0,1]$, the following identity holds:
    \begin{equation}
    \int_0^1 x(1-x)f''(x) dx = f(0) + f(1) - 2
    \end{equation}
    \end{proposition}
    
    \begin{proof}
    A direct calculation, using integration by parts twice, and recalling that $f(x)$ is a probability density function (in particular, it integrates to one) yields
    \begin{equation}
        \begin{aligned}
            \int_0^1 x(1-x)f''(x)dx &= \underbrace{\bigg[x(1-x)f'(x)\bigg]_{0}^1}_{=0} - \int_0^1 (1-2x)f'(x)dx \\
            &= -\bigg[(1-2x)f(x)\bigg]_{0}^1 -2\underbrace{\int_0^1f(x)dx}_{=1} \\
            &= f(0) + f(1) - 2.
        \end{aligned}
    \end{equation}
    \end{proof}
    
    \begin{remark}
        Combining Propositions \ref{prop:massDeviationGeneral} and \ref{prop:massDeviationSpecific}, the total probability mass of the estimator is
        $$ \int_0^1 \ftwo(x)dx = 1 + \frac{h}{2}(f(0) + f(1) - 2) + \mathcal{O}(h^2) $$
        This result indicates that first-order mass conservation depends on the boundary values of the true density. For densities where $f(0) + f(1) > 2$ (for example, ``U-shaped'' distributions), the estimator tends to overestimate the total probability mass. Conversely, for densities where $f(0) + f(1) < 2$ (for example, ``bell-shaped'' distributions vanishing at the boundaries), the mass is underestimated. Note that for distributions where $f(x) \to \infty$ at the boundaries (violating the differentiability assumption), the asymptotic approximation breaks down, potentially leading to larger deviations, as observed in the experiments.
    \end{remark}

    \subsection{Empirical validation from Experiment 1}
        Figure \ref{fig:IntegralError_vs_N} shows the mean absolute deviation from the unit probability mass, $|\int_0^1 \hat{f}_2(x)dx - 1|$, for the test distributions.
    
        As predicted by Proposition \ref{prop:massDeviationGeneral}, the mass error decays asymptotically as the sample size $n$ increases (and bandwidth $h$ decreases). Crucially, for moderate sample sizes ($n \ge 100$), the deviation is consistently small (typically $< 1\%$), confirming that the choice to use the non-normalized estimator has a negligible impact on practical performance compared to the reduction in boundary bias.
    
        \begin{figure}[h!]
            \centering
            \includegraphics[width=\textwidth]{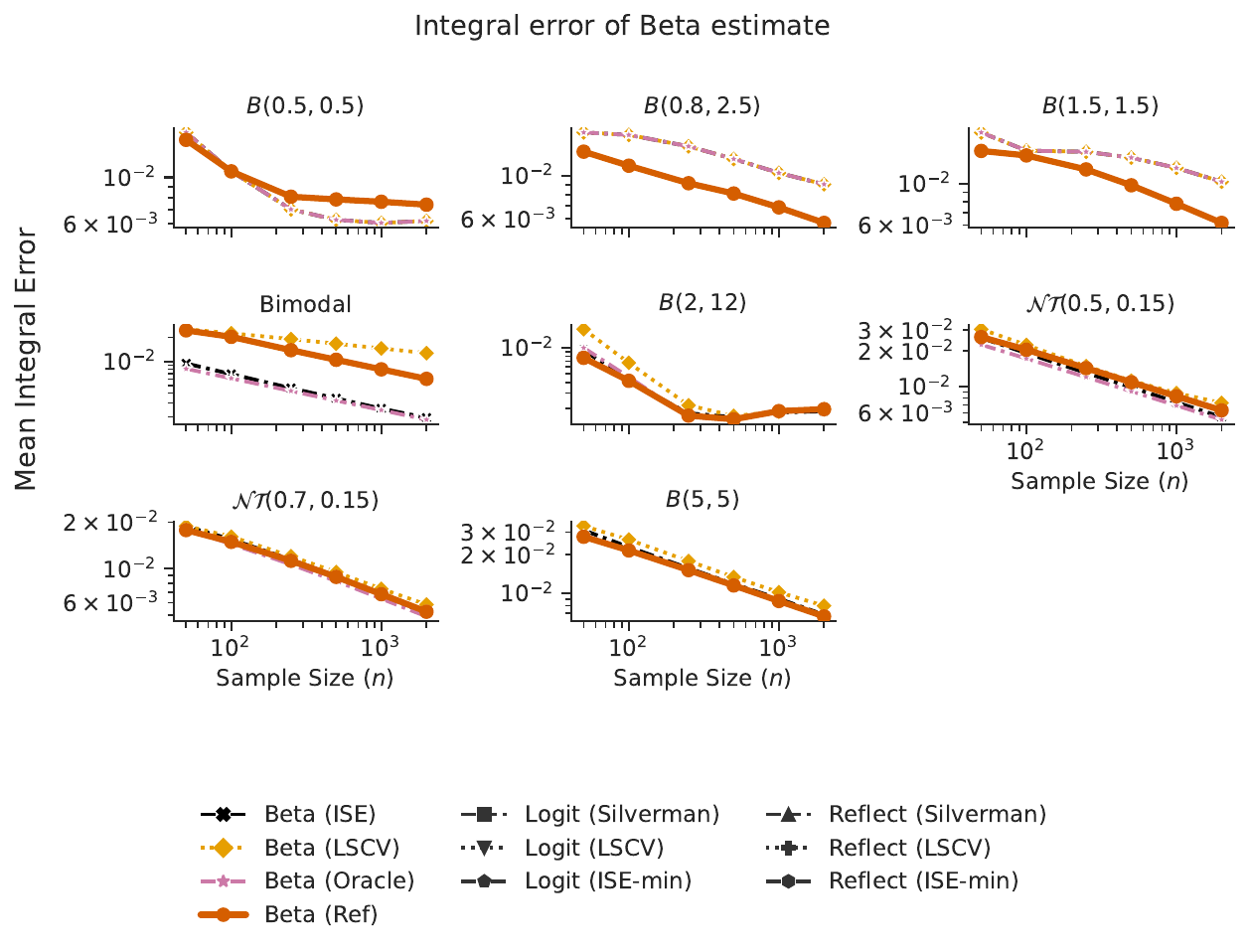}
            \caption{Absolute deviation from unit probability mass $|\int \ftwo(x)dx - 1|$ as a function of sample size $n$. The deviation decays linearly with bandwidth $h$ (and thus with $n$), becoming negligible ($<10^{-2}$) for moderate sample sizes}
            \label{fig:IntegralError_vs_N}
        \end{figure}

\section{Ablation study for the Fallback Rule}\label{app:ablation}
    
    To address the concerns regarding the theoretical motivation for the proposed fallback rule ($\hheur$), we conducted a comprehensive ablation study to justify its specific form. The proposed heuristic explicitly incorporates the sample variance, skewness, and excess kurtosis. To prove that this specific combination is not arbitrary, we evaluated its performance against all partial combinations of these penalty terms: (1) variance only, (2) variance + skewness, and (3) variance + kurtosis.

    The isolated combinations were tested across four archetypal bounded distributions representing different challenges: $B(0.5, 0.5)$ (a heavy-boundary U-shape), $B(0.8, 2.5)$ (an asymmetric J-shape), $B(1.5, 1.5)$ (a symmetric bell shape), and a bimodal mixture distribution (a mixture of $B(10, 30)$ and $B(30, 10)$ with a mixing parameter of 1/2). For each distribution, 1,000 independent trials were simulated across six varying sample sizes ($n=50, 100, 250, 500, 1000, 2000$), totaling 6,000 trials per distribution.

    We report the mean LSCV scores for each form. Because the mean LSCV can occasionally be skewed by extreme optimization failures in nonparametric estimators, we also report robust empirical ``win rates'' (the percentage of trials in which the proposed rule achieved a strictly superior LSCV score compared to the partial baseline) and assess statistical significance using the nonparametric Wilcoxon signed-rank test.

    The results, detailed in Table \ref{tab:ablation} and Figure \ref{fig:ablation}, demonstrate that the full proposed form is the strictly dominant heuristic.

    \begin{table}[h!]
        \centering
        \resizebox{\textwidth}{!}{
            \input{tables/ablation_table}
        }
        \caption{Ablation study of fallback heuristic components across 6,000 trials per distribution. Mean LSCV scores (median in parentheses); lower is better. Bold indicates the best median per distribution. Win rates of the proposed rule. Significance of Wilcoxon signed-rank tests: $^{*}p<0.05$, $^{**}p<0.01$, $^{***}p<0.001$.}
        \label{tab:ablation}
    \end{table}
    
    \begin{figure}[h!]
        \centering
        \includegraphics[width=\textwidth]{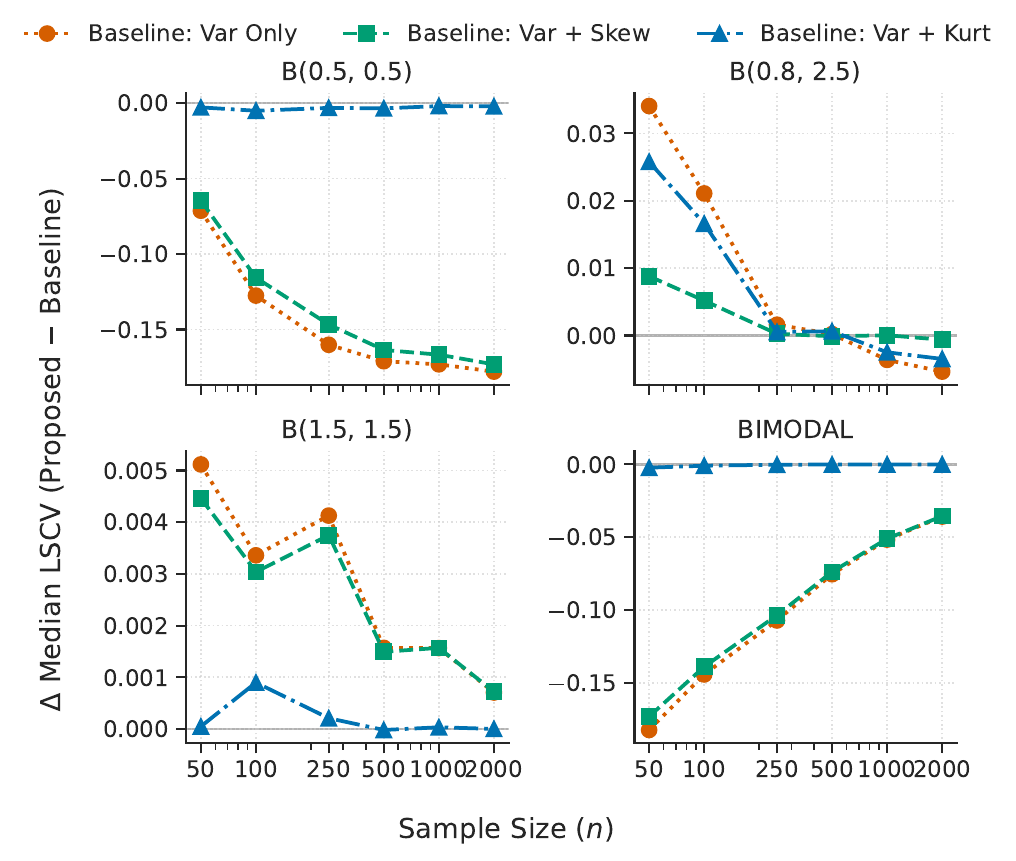}
        \caption{Difference in Median LSCV ($\Delta$) between the proposed rule and partial forms across sample sizes. Values below zero indicate the full proposed rule achieves a better (more negative) score. The full form prevents structural failure on complex distributions (the U-shaped $B(0.5,0.5)$ and Bimodal) with only negligible, asymptotically vanishing trade-offs on simple distributions (J-shaped $B(0.8,2.5)$) and bell shaped $B(1.5,1.5)$}
        \label{fig:ablation}
    \end{figure}
    
    Omitting either the skewness or kurtosis terms leaves the estimator structurally vulnerable to specific distribution shapes. For instance, on the U-shaped $B(0.5, 0.5)$ and Bimodal distributions, the simpler ``Variance Only'' and ``Variance + Skewness'' rules severely underperform. While adding kurtosis (the ``Variance + Kurtosis'' rule) improves performance on these symmetric extremes, it still lacks the necessary asymmetry adjustments.
    
    By incorporating all three terms, the proposed rule explicitly detects these hazards. It achieves massive, statistically significant bandwidth improvements on hard distributions ($p < 0.001$), winning in $>$94\% of trials and yielding a substantially better mean LSCV (-1.6384 on the U-shape and -1.9795 on the Bimodal distribution).
    
    On relatively simple distribution shapes (e.g., $B(1.5, 1.5)$), where boundary accumulation is absent, the higher-order penalty terms remain largely dormant. Although simpler partial rules (such as ``Var+Skew'') may occasionally achieve marginally better efficiency on specific simple asymmetric shapes, such as the J-shaped $B(0.8, 2.5)$, the proposed full rule sacrifices only a microscopic fraction of efficiency in these safe regimes (a difference that vanishes asymptotically as the sample size increases, as confirmed by Figure \ref{fig:ablation}) to prevent the catastrophic structural failures seen on bimodal and U-shaped distributions.

    These results indicate that the specific form of the fallback rule is well-motivated. The ablation study demonstrates that incorporating the full combination of variance, skewness, and kurtosis provides a highly robust heuristic structure. Compared with simpler partial combinations, this full form effectively mitigates the risk of catastrophic boundaries and structural failures, serving as a reliable safeguard across a diverse range of target distribution shapes.
    
\section{Open-Source Software Implementations}\label{app:package}
    To facilitate the practical application of the proposed Hallberg-Szabadváry (HS) bandwidth selector and boundary-corrected beta kernel density estimator, the methods described in this paper have been implemented as open-source software across three major statistical computing environments. The primary goal of these implementations is to provide practitioners with a fast, rule-of-thumb solution for density estimation on bounded intervals that does not rely on computationally expensive or unstable numerical optimization. 

    The implementations are available in the following language ecosystems.
    
    \subsection*{Python implementation (\texttt{beta-kde})}
        We provide a fully documented open-source Python package, \texttt{beta-kde}. The package is available via the Python Package Index (PyPI) and is designed to be API-compatible with standard libraries, such as \texttt{scikit-learn} \citep{scikit-learn,sklearn_api}. 
        The source code of the \texttt{beta-kde} Python package, together with example notebooks, can be found at the following links:
        \begin{itemize}
        \item \href{https://github.com/egonmedhatten/beta-kde}{https://github.com/egonmedhatten/beta-kde}
        \item \href{https://pypi.org/project/beta-kde/}{https://pypi.org/project/beta-kde/}
        \end{itemize}
        This includes efficient implementations of the proposed rule-of-thumb selector, fallback heuristic, and exact LSCV objective functions. By inheriting from \texttt{BaseEstimator}, the package ensures seamless integration with existing machine learning ecosystems, allowing users to leverage standard utilities such as cross-validation strategies and pipeline composition.
    
        Although the primary contribution of this study is the derivation of the bandwidth rule for univariate data, the package also supports practical workflows involving high-dimensional bounded data ($x \in [0,1]^d$). Unlike the standard multivariate KDE, which typically relies on isotropic bandwidths that struggle with bounded hypercubes, our package implements a nonparametric copula strategy. Leveraging Sklar's theorem \citep{sklar1959fonctions}, this approach decomposes the multivariate joint density into univariate marginals and a dependence structure. The package automatically applies the proposed Beta Reference Rule to strictly correct the boundary bias of each univariate marginal. Subsequently, it models the dependence structure (the copula density) using a multivariate product beta kernel estimator on the unit hypercube. This allows practitioners to model bounded multivariate data immediately while strictly respecting the boundary constraints; however, we emphasize that the theoretical analysis of the optimal bandwidth selection for the dependence structure remains a subject for future research.

    \subsection*{R Implementation (\texttt{kdensity})}
        The method is natively integrated into the R statistical computing environment via the standard \href{https://cran.r-project.org/web/packages/kdensity/index.html}{\texttt{kdensity}} package, available on the Comprehensive R Archive Network (CRAN) starting from version 1.2.0. The method is invoked using the \texttt{bw = "HS"} argument. Notably, the HS closed-form reference rule has been adopted as the default bandwidth selector in \texttt{kdensity} whenever a beta kernel is specified (\texttt{kernel = ``beta''}) alongside a uniform or constant parametric start.

    \subsection*{Julia Implementation (\texttt{BetaKDE.jl})}
        For high-performance and scientific computing applications, the estimator is available in the \href{https://github.com/egonmedhatten/BetaKDE.jl}{\texttt{BetaKDE.jl}} package, which can be installed via the official Julia General Registry. The package is optimized for speed and integrates directly with the broader \texttt{JuliaStats} ecosystem.
        
\end{document}

%% file: tables/main_table_lscv.tex
\begin{tabular}{lccc}
\toprule
 & 'Nice' Distributions & 'Bimodal' Distribution & 'Hard' Distributions \\
Method &  &  &  \\
\midrule
\rott & -2.2667 (-1.9134) & -1.9795 (-2.0221) & \textbf{-1.5459 (-1.6411)} \\
\blscvt & -2.2716 (-1.9155)$^{***}$ & -1.9603 (-1.9878)$^{***}$ & -1.5200 (-1.5532)$^{***}$ \\
\biset & -2.2598 (-1.9120)$^{***}$ & -2.0211 (-2.0511)$^{***}$ & - \\
\oraclet & -2.2664 (-1.9139)$^{***}$ & -2.0262 (-2.0526)$^{***}$ & - \\
\lsilvt & -2.2557 (-1.9055)$^{***}$ & -1.8364 (-1.8831)$^{***}$ & -3.7532 (-1.1105)$^{***}$ \\
\llscvt & -2.2610 (-1.9087)$^{***}$ & -2.0287 (-2.0528)$^{***}$ & -2.9911 (-1.0972)$^{***}$ \\
\liset & -2.2482 (-1.9037)$^{***}$ & -2.0174 (-2.0491)$^{***}$ & - \\
\rsilvt & -2.2559 (-1.9102)$^{***}$ & -1.8669 (-1.9122)$^{***}$ & -1.4185 (-1.2950)$^{***}$ \\
\rlscvt & \textbf{-2.2700 (-1.9175)}$^{*}$ & \textbf{-2.0312 (-2.0538)}$^{***}$ & -1.5065 (-1.5345)$^{***}$ \\
\riset & -2.2529 (-1.9112)$^{***}$ & -2.0196 (-2.0502)$^{***}$ & - \\
\bottomrule
\end{tabular}

%% file: tables/main_table_ise.tex
\begin{tabular}{lcc}
\toprule
 & 'Nice' Distributions & 'Bimodal' Distribution \\
Method &  &  \\
\midrule
\rott & 0.0313 (0.0149) & 0.0914 (0.0569) \\
\blscvt & 0.0358 (0.0161)$^{***}$ & 0.1104 (0.0965)$^{***}$ \\
\biset & \textbf{0.0272 (0.0130)}$^{***}$ & \textbf{0.0381 (0.0219)}$^{***}$ \\
\oraclet & 0.0297 (0.0144)$^{***}$ & 0.0403 (0.0231)$^{***}$ \\
\lsilvt & 0.0431 (0.0206)$^{***}$ & 0.2377 (0.1948)$^{***}$ \\
\llscvt & 0.0604 (0.0256)$^{***}$ & 0.0515 (0.0287)$^{***}$ \\
\liset & 0.0383 (0.0182)$^{***}$ & 0.0415 (0.0237)$^{***}$ \\
\rsilvt & 0.0413 (0.0219)$^{***}$ & 0.2067 (0.1674)$^{***}$ \\
\rlscvt & 0.0534 (0.0227)$^{***}$ & 0.0501 (0.0278)$^{***}$ \\
\riset & 0.0333 (0.0157) & 0.0397 (0.0228)$^{***}$ \\
\bottomrule
\end{tabular}

%% file: tables/main_table_time.tex
\begin{tabular}{lccc}
\toprule
 & 'Nice' Distributions & 'Bimodal' Distribution & 'Hard' Distributions \\
Method &  &  &  \\
\midrule
\rott & 0.0001 & 0.0001 & 0.0001 \\
\blscvt & 3.5567 & 3.5616 & 3.6010 \\
\biset & 1.1719 & 1.6634 & - \\
\oraclet & 0.0000 & 0.0000 & - \\
\lsilvt & 0.0001 & 0.0001 & 0.0001 \\
\llscvt & 2.1040 & 2.1839 & 2.1680 \\
\liset & 0.5782 & 0.8709 & - \\
\rsilvt & 0.0001 & 0.0001 & 0.0001 \\
\rlscvt & 21.7347 & 21.6673 & 25.3464 \\
\riset & 0.3862 & 0.7814 & - \\
\bottomrule
\end{tabular}

%% file: tables/experiment_2_table.tex
\begin{tabular}{lcccccc}
\hline
\textbf{Dataset} & \textbf{Method} & \textbf{LSCV Score} & \textbf{Heldout Density} & \textbf{Time (s)} & \textbf{Fallback Rate} \\ \hline
\multirow{6}{*}{\textit{PctKids2Par}} 
 & \rott & -1.4030 & 1.376 (1.380) & 0.0002 & 0\% \\
 & \blscvt & -1.4031 & 1.386 (1.390) & 15.0 & - \\
 & \lsilvt & -1.3691$^{***}$ & 1.316 (1.316)$^{***}$ & 0.0001 & - \\
 & \llscvt & -1.4016 & 1.392 (1.395) & 7.1 & - \\
 & \rsilvt & \textbf{-1.4037} & \textbf{1.394 (1.399)} & $<0.0001$ & - \\
 & \rlscvt & -1.4037 & 1.395 (1.399) & 71.6 & - \\
\hline
\multirow{6}{*}{\textit{PctPopUnderPov}} 
 & \rott & -1.5949 & 1.583 (1.588) & 0.0001 & 100\% \\
 & \blscvt & -1.5913$^{***}$ & 1.563 (1.563)$^{***}$ & 14.3 & - \\
 & \lsilvt & -1.5470$^{***}$ & 1.522 (1.519)$^{***}$ & $<0.0001$ & - \\
 & \llscvt & -1.6018 & 1.596 (1.601) & 7.1 & - \\
 & \rsilvt & -1.5589$^{***}$ & 1.532 (1.530)$^{***}$ & $<0.0001$ & - \\
 & \rlscvt & \textbf{-1.6063} & \textbf{1.603 (1.607)} & 74.1 & - \\
\hline
\multirow{6}{*}{\textit{PctVacantBoarded}} 
 & \rott & -2.1066 & 2.203 (2.201) & 0.0001 & 100\% \\
 & \blscvt & -2.1406$^{***}$ & 2.163 (2.166)$^{***}$ & 13.3 & - \\
 & \lsilvt & -1.5746$^{***}$ & 1.859 (1.856)$^{***}$ & $<0.0001$ & - \\
 & \llscvt & -2.2250 & 2.206 (2.199) & 10.0 & - \\
 & \rsilvt & -2.5611 & 2.525 (2.527) & $<0.0001$ & - \\
 & \rlscvt & \textbf{-2.6565} & \textbf{2.624 (2.622)} & 88.3 & - \\
\hline
\end{tabular}

%% file: tables/ablation_table.tex
\begin{tabular}{lllll}
\toprule
 & B(0.5, 0.5) & B(0.8, 2.5) & B(1.5, 1.5) & BIMODAL \\
Metric &  &  &  &  \\
\midrule
Var Only LSCV & -1.4764 (-1.5095)$^{***}$ & -1.9442 (-1.9552)$^{***}$ & \textbf{-1.0629 (-1.0695)}$^{***}$ & -1.8815 (-1.9334)$^{***}$ \\
Var+Skew LSCV & -1.4880 (-1.5198)$^{***}$ & -1.9410 (-1.9554)$^{***}$ & -1.0626 (-1.0694)$^{***}$ & -1.8846 (-1.9355)$^{***}$ \\
Var+Kurt LSCV & -1.6339 (-1.6751)$^{***}$ & \textbf{-1.9441 (-1.9559)}$^{***}$ & -1.0584 (-1.0680)$^{***}$ & -1.9789 (-2.0226)$^{***}$ \\
Proposed LSCV & \textbf{-1.6384 (-1.6792)} & -1.9387 (-1.9540) & -1.0580 (-1.0680) & \textbf{-1.9795 (-2.0229)} \\
\midrule
Win Rate (vs Var Only) & 94.9\% & 50.6\% & 8.9\% & 96.4\% \\
Win Rate (vs Var+Skew) & 94.6\% & 39.6\% & 8.7\% & 96.4\% \\
Win Rate (vs Var+Kurt) & 91.7\% & 47.8\% & 6.0\% & 96.4\% \\
\bottomrule
\end{tabular}